\documentclass[a4paper,UKenglish]{article}

\usepackage{a4wide}
\usepackage{microtype} 

\usepackage{amsmath}
\usepackage{amssymb}
\usepackage{amsthm}
\usepackage[usenames,dvipsnames]{xcolor}
\usepackage{proof-dashed}
\usepackage{stmaryrd}
\usepackage{wrapfig}
\usepackage{xspace}
\usepackage{todonotes}
\usepackage{breakcites}
\usepackage{url}
\usepackage{hyperref}
\usepackage{calc}

\usepackage{sung-math}
\usepackage{sung-reo}
\usepackage{pifont}
\newcommand{\yes}{\ensuremath{\text{\ding{51}}}}
\newcommand{\no}{\ensuremath{\text{\ding{55}}}}

\newtheorem{example}{Example}
\newtheorem{definition}{Definition}
\newtheorem{theorem}{Theorem}
\newtheorem{remark}{Remark}
\newtheorem{corollary}{Corollary}

\newcommand{\m}[1]{\ensuremath{\mathtt{#1}}}
\newcommand{\pid}[1]{\m{#1}}
\newcommand{\port}[1]{\m{#1}}
\newcommand{\pidp}[2]{\pid{#1}.\port{#2}}

\newcommand{\code}[1]{\texttt{\upshape #1}}
\newcommand{\nil}{\boldsymbol 0}
\newcommand{\com}[2]{#1\;\code{-\hspace{-0.3mm}>}\;#2}
\newcommand{\valcom}[4]{\com{#1 \carr{#2}}{#3}.{#4}}
\newcommand{\gencom}{\valcom{\pid p}{e}{\pid q}{x}}
\newcommand{\genrecv}{\pid q.x?v}
\newcommand{\genchoice}{\pid q[\ell]}

\newcommand{\sel}[3]{\com{#1}{#2 [#3]}}
\newcommand{\gensel}{\sel{\pid p}{\pid q}{\ell}}

\newcommand{\cond}[3]{\m{if}\, #1 \, \m{then} \, #2 \, \m{else} \, #3}
\newcommand{\gencond}{\cond{\pid p.e}{C_1}{C_2}}
\newcommand{\rec}[3]{\m{def} \, #1  =  #2 \, \m{in} \, #3}
\newcommand{\genrec}{\rec{X}{C_2}{C_1}}
\newcommand{\gencall}{X}

\newcommand{\rname}[2]{\ensuremath{\left\lfloor\mbox{{#1}$|${#2}}\right\rceil}}

\newcommand{\precongr}{\mathrel{\preceq}}
\newcommand{\pn}{\mathop{\mathsf{pn}}}
\newcommand{\carr}[1]{\langle #1  \rangle}
\newcommand{\psend}[2]{{#1}!\carr{#2}}
\newcommand{\precv}[2]{#1?#2}
\newcommand{\psel}[2]{{#1}\oplus#2}
\newcommand{\pbranch}[2]{{#1}\&{#2}}
\newcommand{\ppar}{\mathrel{\boldsymbol{|}}}
\newcommand{\proc}[3]{#1 \triangleright_{#2} #3}
\newcommand{\wtil}{\widetilde}
\newcommand{\eval}[2]{\mathrel{\downarrow^{#1}_{#2}}}
\newcommand{\epp}[2]{[\![#1]\!]_{#2}}
\newcommand{\conn}{\mathcal G} 
\newcommand{\cconn}{\mathcal C}
\newcommand{\astate}{\mathcal A} 

\newcommand{\toC}{\redc{\cconn}}
\newcommand{\toCEPP}{\redc{\epp{\conn}{}}}

\newcommand{\mem}{\m{var}}
\newcommand{\redc}[1]{\rightsquigarrow_{#1}}
\newcommand{\redM}{\redc{\mathcal G}}

\newcommand{\med}[1]{\mathsf{#1}}
\newlength {\trlabelwidth}
\newcommand{\tr}[1][]{%
	\IfStrEq{\detokenize{#1}}{}{%
		{\xrightarrow{\hbox to 8pt {}}}%
	}{%
		\settowidth{\trlabelwidth}{\ensuremath{#1}}%
		\addtolength{\trlabelwidth}{2pt}%
		\ifdim\trlabelwidth < 8pt
			\xrightarrow{\hbox to 8pt {\scriptsize\hfil$\smash{#1}$\hfil}}%
		\else
			\xrightarrow{\hspace{1pt}\smash{#1}\hspace{1pt}}%
		\fi
	}
}
\newcommand{\thru}{\,\m{thru}\,}
\newcommand{\wwedge}{\mbox{ and }}

\newcommand{\dcform}{\Phi}
\newcommand{\flow}{\mathbin{\rightarrow}}
\newcommand{\nomid}{\hphantom{{}\wedge{}}}
\newcommand{\til}{\tilde}

\title{Connectors meet Choreographies}

\author{Farhad Arbab$^1$
  \and Lu\'is Cruz-Filipe$^2$
  \and Sung-Shik Jongmans$^{3,4}$
  \and Fabrizio Montesi$^2$}
\date{${}^1$ Centrum Wiskunde \& Informatica, the Netherlands\\
  ${}^2$ Dept.\ Mathematics and Computer Science, Univ.\ Southern Denmark\\
  ${}^3$ Dept.\ Computer Science, Open University of the Netherlands\\
  ${}^4$ Dept.\ Computing, Imperial College London}

\begin{document}

\maketitle

\begin{abstract}
We present Cho-Reo-graphies (CR), a new language model that unites two powerful programming paradigms for concurrent software based on communicating processes: Choreographic Programming and Exogenous Coordination.
In CR, programmers specify the desired communications among processes using a \emph{choreography},
and define how communications should be concretely animated by \emph{connectors} given as constraint automata (e.g., synchronous barriers and asynchronous multi-casts).
CR is the first choreography calculus where different communication semantics
(determined by connectors) can be freely mixed; since connectors are user-defined, CR also supports many communication semantics that were previously unavailable for choreographies.
We develop a static analysis that guarantees that a choreography in CR and its user-defined connectors are compatible, define a compiler from choreographies to a process calculus based on connectors, and prove that compatibility guarantees deadlock-freedom
of the compiled process implementations.
\end{abstract}

\section{Introduction}
\label{sec:intro}

\paragraph{Background.}
Programming concurrent software is hard: given the specification of the desired
interactions among processes, implementation is error-prone even for experienced
programmers \cite{LLLG16,LPSZ08}.
The challenge of concurrent programming has driven decades of research on new
programming models. A particularly fruitful idea was to provide a native
language abstraction for interaction, rather than modelling it as a side-effect.
Two research lines in particular were built upon this idea, but following very
different directions.

The first research line is that on \emph{Choreographic Programming}~\cite{M13:phd}.
Central to Choreographic Programming is the \emph{choreography}, a programming
artefact that specifies a concurrent system in terms of the interactions among
its constituent processes, by using an ``Alice and Bob'' notation that disallows
writing mismatched I/O actions (e.g., a send without a corresponding receive).
Through \emph{EndPoint Projection} (EPP), choreographies can be compiled to
faithful and deadlock-free implementations in process languages~\cite{CM13}.

Previous work studied models for choreographies with different interaction
semantics (e.g., synchronous, asynchronous, multi-cast). Common to all
existing models of Choreographic Programming
is the fact that the nature of interactions
is ``hardcoded'': each model proposes new syntax and semantics, so results have
to be proven from scratch every time and cannot be combined. For example, none
of the existing models supports mixing synchronous with asynchronous
interactions within the same choreography. This is a serious limitation, which
raises the challenge of finding a unifying framework. Furthermore, there is still no
indication of how other interesting interaction semantics (e.g., barriers) can
be introduced to choreographies.

The second research line is that on \emph{Exogenous
Coordination}~\cite{arbab1996iwim,arbab1998you,lau2005exogenous}, where
interaction protocols and process code are developed separately.
Process programs are then modularly composed with
protocols, given as \emph{connectors}, which dictate interactions
by accepting/offering messages from/to processes. Connectors offer an
elegant way of programming different semantics for interactions, starting from
basic instances (e.g., synchronous and asynchronous point-to-point channels) and
then composing them to create more sophisticated protocols.

The limitation of Exogenous Coordination is that we do not have a
global system view when composing processes with connectors: programmers cannot
define the intended flow of information among processes when a connector is used
many times, or if many connectors are used in different parts of a system.
Hence, incompatibilities between process code and connectors may cause
deadlocks, which cannot happen in Choreographic Programming.
The two paradigms thus have complementary strengths.

\paragraph{Contribution.}
We develop the first integration of the best aspects of Choreographic
Programming and Exogenous Coordination:
a new calculus of choreographies, called \emph{Cho-Reo-graphies}
(CR), whereby the process interactions specified in a choreography are
animated by arbitrary, user-defined, connectors based on the Exogenous
Coordination language Reo. CR allows for mixing different connectors in the
same choreography, making it for the first time possible to write choreographies
where different interactions can have different communication semantics.
Furthermore, by tapping into the expressivity of connectors, we can endow
choreographies with hitherto unexplored communication semantics, such as
alternators or barriers. This makes CR more expressive and a generalisation of
existing models of Choreographic Programming. Through EPP, choreographies can be
compiled into concurrent implementations in a process language. In these
implementations, the same connectors as in the original choreography animate
interactions among processes. We show that these processes are deadlock-free,
provided that the original choreography is compatible with the connectors.

Mixing different communication semantics (i.e., connectors)
in choreographies produces
new challenges of formalisation and decidability. Solving these is our main technical contribution.
On the formalisation front, we have to balance expressiveness and
comprehensibility: the formal semantics of the calculus should be easy enough to
explain and understand, without sacrificing the expressive power of
connectors. We address this challenge with a new labelled reduction semantics
for choreographic interactions, where labels act as the interface with
connectors.
On the decidability front, we prove that (1) deadlock-freedom is generally undecidable in our calculus, but
(2) we can establish deadlock-freedom for a large subset of the language. Our
proof of (1) shows that undecidability is a direct consequence of the expressive
power that connectors introduce to choreographies. We address (2) by designing
a new decidable static analysis (compatibility) that is made possible by
the careful design of labelled reductions for interactions.

By leveraging existing work, distributed implementations of Reo connectors can
be automatically generated and deployed on distributed systems (e.g., in
Scala~\cite{proencca2011synchronous,proencca2012dreams} or Java~\cite{jongmans2015partially}); our approach is compatible with this work.
As such, a practical tool based on CR can
in principle be built on top of existing code
generators for distributed implementations of Reo connectors.
In general, we believe that our results represent the beginning of an interesting research
line on concurrent programming; we discuss future directions in \S~\ref{sec:conclusions}.

\paragraph{Structure.}
We motivate our work with an example in \S~\ref{sec:example}, and introduce
connectors in~\S~\ref{sec:reo}.
CR is described in \S~\ref{sec:chor} together with results on deadlock-freedom.
In \S~\ref{sec:processes} we introduce the target process language for EndPoint
Projection and prove correctness of synthesised process implementations.
We conclude in \S~\ref{sec:conclusions}, discussing directions in
which this work can be extended.
\S~\ref{sec:conclusions} concludes, discussing directions in which this work can
be extended.

\paragraph{Related work.}
We already covered the main references to previous work and how it falls short of serving our aim. We briefly recap related work.

Choreographies have been studied in different settings with fixed communication semantics, including synchronous~\cite{CHY12,CM16b,LGMZ08}, asynchronous point-to-point~\cite{CM13,CM17:ice,G16,MY13}, one-to-many~\cite{CMSY16,CDYP16},
and many-to-one~\cite{CLMSW16,LMMNSVY15}.
In our model, these semantics are simply instances of what we can do.
But since we tap into the generality of Reo connectors, we can also do more (e.g.,
we illustrate how to use barriers).
CR is not the first model where choreographies may deadlock: this is common for
settings with realistic communication semantics (e.g.,
\cite{CM16:pc_long,LNN16}).
However, it is the first model with arbitrary communication semantics.
Our development also extends the line of work on out-of-order execution for
choreographies, initiated in~\cite{CM13}, where non-interfering interactions may proceed concurrently. This style allows for more safe behaviours in the semantics of choreographies (by
swapping non-interfering communications), which the programmer gets for free (concurrency is
inferred). As in many other works, out-of-order execution also simplifies our syntax: we do
not need to provide for
a parallel operator in choreographies, since most parallel behaviour is
already captured by out-of-order execution (cf.~\cite{CM13} for details).
As in previous work, CR supports asynchronous behaviour without requiring the programmer to
reason about it in choreographies: communications are still specified as atomic interactions, which
may be asynchronously reduced in a safe way (cf.~\cite{CM13,DY13,MY13}).

In Exogenous Coordination, interaction protocols for communicating processes,
called connectors, are programmed separately from the internal code of each process. This enables a
compositional approach for developing protocols, where complex protocols can be built by
assembling simpler ones. 
Exogenous Coordination has been studied extensively over the last two
decades~\cite{arbab1996iwim,arbab1998you,lau2005exogenous}.
Examples of models of Exogenous Coordination are the algebras of
connectors~\cite{bliudze2008algebra,bliudze2010causal}, the algebra of stateless connectors~\cite{bruni2006basic}, and
constraint automata~\cite{BSAR06}; examples of languages are (interactions in)
BIP~\cite{basu2011rigorous,basu2006modeling}, Ptolemy~\cite{buck1994ptolemy,Ptolemaeus:14:SystemDesign}, and
Reo~\cite{Arb04,Arb11}.

\section{Motivating Example \& Approach}
\label{sec:example}

We present an example to introduce the concept of choreographies and the problem
we are interested in studying. This example will be used as running example
throughout the whole article.

\begin{example}[Book sale]
\label{ex:motivating}
Alice ($\pid a$) wants to buy a book from seller Carol ($\pid c$), facilitated by a bank ($\pid b$) and a shipper ($\pid s$).
First, Alice sends the title of the book to Carol.
Carol then replies to Alice with the price of the book.
If Alice is happy with the price, she notifies Carol, the bank, and the shipper that the purchase proceeds.
Alice subsequently sends the money to the bank (who transfers it to Carol's account), and Carol sends the book to the shipper (who dispatches it to Alice).
The choreography for this scenario looks as follows:
\setlength{\belowdisplayskip}{0pt}
\begin{align*}
	1.\quad & \com{\pid a\langle\mathit{title}\rangle}{\pid c};\
	\com{\pid c\langle\mathit{price}\rangle}{\pid a};
\\
	2.\quad & \m{if} \ {\pid a.\mathit{happy}} \
	\m{then}\ (\ \sel{\pid a}{\pid c}{\mathit{ok}};\
        \sel{\pid a}{\pid b}{\mathit{ok}};\ \sel{\pid a}{\pid s}{\mathit{ok}};
	\com{\pid a\langle\mathit{money}\rangle}{\pid b};\
        \com{\pid c\langle\mathit{book}\rangle}{\pid s}\ )
\\
        &\phantom{\m{if} \ {\pid a.\mathit{happy}} \ }
	\m{else}\ (\ \sel{\pid a}{\pid c}{\mathit{ko}};\
	\sel{\pid a}{\pid b}{\mathit{ko}};\ \sel{\pid a}{\pid s}{\mathit{ko}}\ )
        \hspace*{12em}\qed
\end{align*}
\end{example}

In previous choreography models, the nature of the interactions is fixed: depending on the model,
these interactions are either all synchronous or all asynchronous.
This is not flexible enough, because requirements may be different for each interaction.

\begin{example}[Book sale]
\label{ex:motivating+reqs}
The book sale scenario has the following requirements:

	\begin{itemize}
		\item
		Because there are no strict timing constraints, it is reasonable for Alice and Carol to communicate asynchronously in line 1.

		\item
		Because the same label ($\mathit{ok}$ or $\mathit{ko}$) is sent from Alice to Carol, the bank, and the shipper, it makes sense to combine these communications in a multi-cast.

		\item
		It is better for Alice to send her money to the bank as late as possible (e.g., because she receives interest on her money).
		Therefore, Alice does not want to send her money until she knows the others have received her $\mathit{ok}$-label.
		Thus, the multi-cast of this label should be synchronous (i.e., handshake between Alice, Carol, the bank, and the shipper).

		\item
		Alice and Carol may not trust each other: Alice does not want to
                send money before Carol has sent the book, and vice versa.
		To resolve this impasse,
                Alice and Carol should synchronously (barrier-like) send
                money and book, to ensure each of them holds her end of the
                bargain; the receives by the bank and the shipper may
                subsequently proceed asynchronously.
		\qed
	\end{itemize}
\end{example}

\noindent
As we show in the next sections, our model is powerful enough to express all
these interactions (and more).
The key idea is to tag interactions with the name of the particular connector through which they transpire.
For instance, a value communication from $\pid p$ to $\pid q$ through a synchronous channel $\med{sync}$ is expressed as $\com{\pid p}{\pid q} \thru \med{sync}$, while the same communication through an asynchronous channel $\med{async}$ is expressed as $\com{\pid p}{\pid q} \thru \med{async}$.
If all interactions in a choreography transpire through the same (type of) connector, as in all existing choreography approaches, tags become redundant and can be omitted.

\section{Reo and Constraint Automata}
\label{sec:reo}

We view processes in a concurrent system as black boxes with interfaces consisting of \emph{ports}.
For a process to send (receive) a message to (from) another process, it performs an output action (input action) on one of its own output ports (input ports), but without specifying a receiver (sender).
Instead, a separate connector, connected to the ports of the processes, decides how messages flow from senders' output ports to receivers' input ports.

\emph{Reo}~\cite{Arb04,Arb11} is a graphical, dataflow-inspired language for connectors. As we focus primarily on semantics in the rest of this paper, we only present Reo's semantic formalism: \emph{constraint automata}~\cite{BSAR06}.
The transitions of a constraint automaton model possible synchronous message flows (i.e., interactions) through a connector in a particular state.
Constraint automata are parametrised over a language of constraints $\dcform$ to specify transition labels.
Depending on the required level of expressiveness, different instantiations of $\dcform$ may be considered.
In this work, we consider $\dcform_{P,M}$ to be a language of constraints over two sets $P$ and $M$, of ports
and \emph{memory cells} (storage space local to a connector).
Constraints are finite sets of formulas of the form $p_1 \flow p_2$ (port $p_1$ passes a message to port $p_2$),
$m_1 \flow m_2$ (cell $m_1$ passes a message to cell $m_2$), $p \flow m$ (port $p$ passes a message to cell $m$), and
$m \flow p$ (cell $m$ passes a message to port $p$).\footnote{We can formalise $\dcform$ as an equational theory as in~\cite{BSAR06,Jon16}, but this is unnecessary for our development.}
Furthermore, we require that all terms on the right-hand side of formulas in the same constraint in a transition label be
distinct, to ensure that every port\slash cell is assigned a unique message (e.g., $p_1 \flow m \wedge p_2 \flow m$ is forbidden, but $p \flow m_1 \wedge p \flow m_2$ is allowed).

\begin{definition}
A constraint automaton is a tuple $(S, P, M , \tr , s_0 , \mu_0)$, where $S$ is
a finite set of states, $P$ is a finite set of ports, $M$ is a finite set of
memory cells, $\tr \subseteq S \times \dcform_{P,M} \times S$ is a transition
relation, $s_0 \in S$ is an initial state, and $\mu_0$ is an initial memory
snapshot mapping the memory cells in $M$ to their initial content.
\end{definition}

A transition $(s , \phi , s') \in \tr$, which we write $s \tr[\phi] s'$ for
short, means that, from state $s$, a subset of the ports in $P$ can interact
according to $\phi$ and the automaton goes to state $s'$.
Constraints in~$\dcform_{P,M}$ also control the evolution of memory snapshots as transitions are made.
Instead of formalizing this separately (e.g., in terms of runs and languages~\cite{Jon16}), we combine it directly in
the semantics of our choreographies, in \S~\ref{sec:chor}.

In this work, we restrict ourselves to a subset of constraint automata that satisfy two additional assumptions on their transitions.
\label{constraints}
First, the occurrences of each port in transition labels are either all as
sender (the port occurs only on the left-hand side of constraints) or all
as receiver (it occurs only on the right-hand side of constraints); this constraint is imposed on the whole automaton.
For instance, a constraint automaton cannot have two transitions where one is
labelled by $p_1\flow p_2$ and the other by $p_2\flow p_3$.
Secondly, the transition relation of each automaton is deterministic on the ports used, i.e., for any given state $s$, if there
are two distinct transitions $s\tr[\phi_1] s_1$ and $s\tr[\phi_2] s_2$, then the sets of ports used in $\phi_1$ and $\phi_2$ must
be distinct (but they can overlap, or one be a strict subset of the other).
The first assumption simplifies defining the semantics of choreographies in \S~\ref{sec:chor}; the second
assumption ensures that this semantics is deterministic, in line with previous work on
choreographies.

Port automata are minimalistic, meaning that a port automaton defines only the orders in which ports may be used.
On the contrary, a port automaton says nothing about the actual messages being exchanged through ports.
Notably, a port automaton does not define which particular message (sent through an output port by a sender) a receiver receives (through an input port).
Port automata, thus, \emph{underspecify} connector behaviour.

\begin{figure}[t]
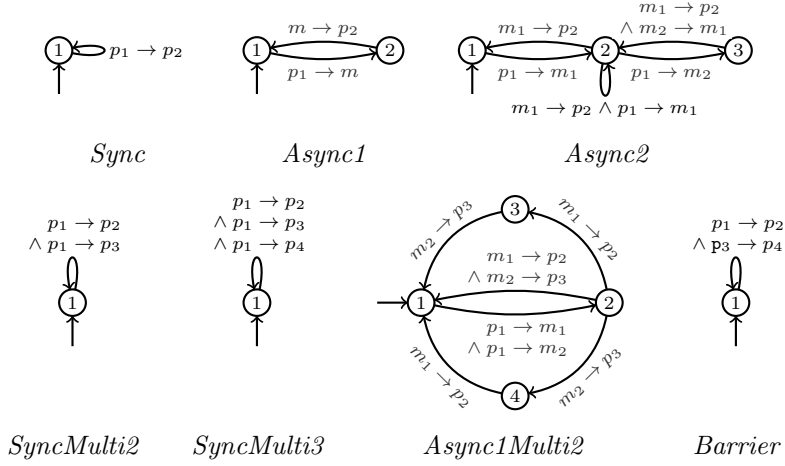

	\centering
	\begin{tabular}{@{}c@{\qquad}c@{\qquad}c@{}}
		\tikzautomone[1]{p_1 \flow p_2}{}{}{}{}{}{init}{}
	&	\begin{tikzautom}
			\state{S1}{0,0}{1}
			\state{S2}{right of=S1}{2}
			\trans{S1}{S2}{from16,to9}{below}{p_1 \flow m}
			\trans{S2}{S1}{from8,to1}{above}{m \flow p_2}
			\init{below of}{S1}
		\end{tikzautom}
	&	\begin{tikzautom}
			\state{S1}{0,0}{1}
			\state{S2}{right of=S1}{2}
			\state{S3}{right of=S2}{3}
			\trans{S1}{S2}{from16,to9}{below}{p_1 \flow m_1}
			\trans{S2}{S1}{from8,to1}{above}{m_1 \flow p_2}
			\trans{S2}{S2}{from12,to13}{below}{\LINES[l]{m_1 \flow p_2 \wedge p_1 \flow m_1}}
			\trans{S2}{S3}{from16,to9}{below}{p_1 \flow m_2}
			\init{below of}{S1}
			\trans{S3}{S2}{from8,to1}{above}{\LINES[l]{\nomid m_1 \flow p_2 \\ {} \wedge m_2 \flow m_1}}
		\end{tikzautom}
	\\	\vspace{-5.5pt}
	\\	\reo{Sync} 
	&	\reo{Async1} 
	&	\reo{Async2} 
	\end{tabular}

	\vspace{5.5pt}
	\begin{tabular}{@{}c@{\qquad}c@{\qquad}c@{\qquad}c@{}}
		\tikzautomone[1]{}{}{\LINES[l]{\nomid p_1 \flow p_2 \\ {} \wedge p_1 \flow p_3}}{}{}{}{init}{}
	&	\tikzautomone[1]{}{}{\LINES[l]{\nomid p_1 \flow p_2 \\ {} \wedge p_1 \flow p_3 \\ {} \wedge p_1 \flow p_4}}{}{}{}{init}{}
	&	\begin{tikzautom}
			\state{S1}{0,0}{1}
			\state{S3a}{above right of=S1}{3}
			\state{S3b}{below right of=S1}{4}
			\state{S2}{below right of=S3a}{2}
			\trans{S1}{S2}{from16,to9}{below}{\LINES[l]{\nomid p_1 \flow m_1 \\ {} \wedge p_1 \flow m_2}}
			\trans{S2}{S1}{from8,to1}{above}{\LINES[l]{\nomid m_1 \flow p_2 \\ {} \wedge m_2 \flow p_3}}
			\trans{S2}{S3a}{from5,to16}{above,sloped}{m_1 \flow p_2}
			\trans{S2}{S3b}{from12,to1}{below,sloped}{m_2 \flow p_3}
			\trans{S3a}{S1}{from9,to4}{above, sloped}{m_2 \flow p_3}
			\trans{S3b}{S1}{from8,to13}{below, sloped}{m_1 \flow p_2}
			\init{left of}{S1}
		\end{tikzautom}
	&	\tikzautomone[1]{}{}{\LINES[l]{\nomid p_1 \flow p_2 \\ {} \wedge \pid p_3 \flow p_4}}{}{}{}{init}{}
	\\	\vspace{-5.5pt}
	\\	\reo{SyncMulti2} 
	&	\reo{SyncMulti3} 
	&	\reo{Async1Multi2} 
	&	\reo{Barrier} 
	\end{tabular}

	\caption{Example constraint automata.}
	\label{fig:aut}
\end{figure}

\begin{example}
Fig.~\ref{fig:aut} shows example constraint automata for useful connectors.
\reo{Sync} models a synchronous channel, which indefinitely lets two processes
synchronously send and receive a message through ports $p_1$ and $p_2$.
\reo{Async1} models an asynchronous channel with a 1-capacity buffer (using a memory cell~$m$).
Indefinitely, first, this connector lets a process send a message through port $p_1$; subsequently, it lets a process receive the message through port $p_2$.
\reo{Async2} models an asynchronous channel with a 2-capacity buffer.
\reo{SyncMulti2} and \reo{SyncMulti3} model synchronous multi-cast connectors for two and three receivers.
Indefinitely, this connector lets three processes synchronously send and receive a message, from port $p_1$ to ports $p_2$ and $p_3$.
\reo{Async1Multi2} models an asynchronous multi-cast connector for two receivers.
Indefinitely, first, the connector lets a process send a message (lower middle transition); then, it lets two processes receive the message, either one after the other (top or bottom transitions) or simultaneously (upper middle transition).
\reo{Barrier} models a barrier send/receive connector.
Indefinitely, this connector lets two processes synchronously send and receive a message through ports $p_1$ and $p_2$, while it synchronously lets processes synchronously send and receive another message through ports $p_3$ and $p_4$.
\qed
\end{example}

Choreography programmers can model connectors by explicitly defining constraint automata.
Alternatively, programmers can model connectors by \emph{composing} constraint automata from basic primitive constraint automata, using a synchronous product operator~\cite{BSAR06}.
A significant advantage of this latter approach is that there exists an intuitive and user-friendly graphical syntax for constraint automaton product expressions.
In this graphical syntax, called \emph{Reo}~\cite{Arb04,Arb11}, programmers draw constraint automaton product expressions as data-flow graphs between (ports of) processes.
As an example, Figure~\ref{fig:reo} shows Reo connectors for the constraint automata in Figure~\ref{fig:aut}.
Reo conveniently hides from choreography programmers the intimidating act of explicitly composing constraint automata, without sacrificing generality:
Reo is complete for constraint automata (under the instantiation of $\dcform$ considered in this paper), meaning that every constraint automaton can be expressed as a Reo connector~\cite{DBLP:conf/coordination/ArbabBBRS05,DBLP:journals/fuin/BaierKK14}.
Moreover, tooling exists to animate flows of messages through Reo connectors (\url{http://reo.project.cwi.nl}).

\begin{figure}[t]
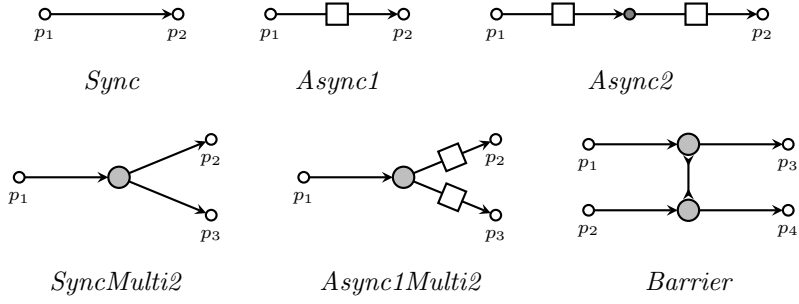

	\centering
	\begin{tabular}{@{}c@{\qquad}c@{\qquad}c@{}}
		\begin{tikzcirc}
			\port{P1}{0,0}{below}{p_1}
			\port{P2}{right of=P1}{below}{p_2}
			\sync{P1}{P2}
		\end{tikzcirc}
	&	\begin{tikzcirc}
			\port{P1}{0,0}{below}{p_1}
			\port{P2}{right of=P1}{below}{p_2}
			\fifoo{P1}{P2}{}
		\end{tikzcirc}
	&	\begin{tikzcirc}
			\port{P1}{0,0}{below}{p_1}
			\port[intern]{P2}{right of=P1}{below}{}
			\port{P3}{right of=P2}{below}{p_2}
			\fifoo{P1}{P2}{}
			\fifoo{P2}{P3}{}
		\end{tikzcirc}
	\\	\vspace{-5.5pt}
	\\	\reo{Sync} 
	&	\reo{Async1} 
	&	\reo{Async2} 
	\end{tabular}

	\vspace{5.5pt}

	\begin{tabular}{@{}c@{\qquad}c@{\qquad}c@{}}
		\begin{tikzcirc}
			\reonode{N}{0,0}{}{}
			\port{P1}{left of=N,node distance=.75*\dist cm}{below}{p_1}
			\port{P2}{$(N)+(22.5:.75*\dist cm)$}{below}{p_2}
			\port{P3}{$(N)+(-22.5:.75*\dist cm)$}{below}{p_3}
			\sync{P1}{N}
			\sync{N}{P2}
			\sync{N}{P3}
		\end{tikzcirc}
	&	\begin{tikzcirc}
			\reonode{N}{0,0}{}{}
			\port{P1}{left of=N,node distance=.75*\dist cm}{below}{p_1}
			\port{P2}{$(N)+(22.5:.75*\dist cm)$}{below}{p_2}
			\port{P3}{$(N)+(-22.5:.75*\dist cm)$}{below}{p_3}
			\sync{P1}{N}
			\fifoo{N}{P2}
			\fifoo{N}{P3}
		\end{tikzcirc}
	&	\begin{tikzcirc}
			\coordinate (X) at (0,0);
			\reonode{N1}{$(X)+(0,.25*\dist cm)$}{}{}
			\reonode{N2}{below of=N1,node distance=.5*\dist cm}{}{}
			\port{P1}{left of=N1,node distance=.75*\dist cm}{below}{p_1}
			\port{P2}{left of=N2,node distance=.75*\dist cm}{below}{p_2}
			\port{P3}{right of=N1,node distance=.75*\dist cm}{below}{p_3}
			\port{P4}{right of=N2,node distance=.75*\dist cm}{below}{p_4}
			\syncdrain{N1}{N2}
			\sync{P1}{N1}
			\sync{P2}{N2}
			\sync{N1}{P3}
			\sync{N2}{P4}
		\end{tikzcirc}
	\\	\vspace{-5.5pt}
	\\	\reo{SyncMulti2} 
	&	\reo{Async1Multi2} 
	&	\reo{Barrier} 
	\end{tabular}

	\caption{Example Reo connectors for the constraint automata in Fig.~\ref{fig:aut}.}
	\label{fig:reo}
\end{figure}

We use (constraint) automata and (Reo) connectors interchangeably: at the semantics level, we use automata, while at the syntax level, we use connectors.

\section{Cho-Reo-graphies}
\label{sec:chor}


We now present \emph{Cho-Reo-graphies} (CR): our choreography calculus that combines choreographies with Reo connectors.
A choreography describes the behaviour of a set of processes.
For simplicity, we assume that values are untyped; treating value types is straightforward and analogous to~\cite{CHY12,CM13,CM16:pc_long}.

\begin{figure}[t]
\setlength{\abovedisplayskip}{0pt}
\setlength{\belowdisplayskip}{0pt}
\begin{align*}
C & ::= \tilde\eta\thru\gamma;C \mid \gencond \mid \genrec \mid \gencall \mid \nil
\\
\eta & ::= \gencom \mid \gensel \mid \fbox{\strut$\genrecv$} \mid \fbox{\strut$\genchoice$}
\end{align*}
\caption{Cho-Reo-graphies, syntax. The boxed terms are runtime terms, necessary for defining
    the semantics, but which are not meant to be used by programmers.}
\label{fig:cc_syntax}
\end{figure}

The syntax of CR is displayed in Fig.~\ref{fig:cc_syntax}.
We range over choreographies using $C, C', \ldots$;
over processes using $\pid p, \pid q, \ldots$;
over interactions using $\eta, \eta', \ldots$ (and over sets of them using $\tilde \eta, \tilde \eta', \tilde \eta_1, \ldots$);
over connector names using $\gamma, \gamma', \ldots$;
over procedure names using $X$;
over (side-effect free) expressions using $e, e', \ldots$;
over values using $v, v', \ldots$ ;
and over selection labels using $\ell, \ell', \ldots$
Each process $\pid p$ 
owns a finite set of local variables $\mem_{\pid p}=\{x_1,\ldots,x_n\}$, and
expressions are assumed to be inductively defined including $\mem_{\pid p}$ as base cases.

We comment briefly on label selections.
It is standard practice in Choreographic Programming to distinguish value communications, which are
used to exchange data, from label selections, which are used to propagate decisions regarding
control flow. In the choreography in Ex.~\ref{ex:motivating}, Alice uses label selections
($\mathit{ok}$ or $\mathit{ko}$) to communicate her choice of whether to buy the book or not to
Carol and the bank. Although this can be encoded using value communications~\cite{CM16b}, it is
useful to distinguish them, as they are usually treated differently in implementations. Also,
as discussed in \S~\ref{sec:processes}, label selections are instrumental in generating process
implementations automatically.
The syntax of labels is unspecified.

We write $\com{\pid p\carr e}{\{\pid q_1.x_1,\ldots,\pid q_n.x_n\}}$ (value multicast) and $\sel{\pid p}{\{\pid q_1,\ldots,\pid
q_n\}}\ell$ (label multicast)
to abbreviate $\{\valcom{\pid p}e{\pid q_1}{x_1},\ldots,\valcom{\pid p}e{\pid q_n}{x_n}\}$ and
$\{\sel{\pid p}{\pid q_1}\ell,\ldots,\sel{\pid p}{\pid q_n}\ell\}$, respectively.
Also, we write $\com{\pid p\carr x}{\pid q}$ for $\com{\pid p\carr x}{\pid q.\mathit{x}}$ (i.e., if sender $\pid p$ sends the value in local variable $x$, and if receiver $\pid q$ stores the received value in a local variable with the same name, we omit the name).
If $\tilde \eta$ is a singleton, we omit curly braces.

\begin{example}[Book sale]
\label{ex:motivating+cr}
The choreography
for our running example (Ex.~\ref{ex:motivating}) can be written as follows in CR.
\setlength{\belowdisplayskip}{0pt}
\begin{align*}
	1.\quad & \com{\pid a\langle\mathit{title}\rangle}{\pid c}\thru\med{a2c};
\\	2.\quad & \com{\pid c\langle\mathit{price}\rangle}{\pid a}\thru\med{c2a};
\\	3.\quad & \m{if} \ {\pid a.\mathit{happy}}\ \m{then}\ (\ \sel{\pid a}{\{ \pid c, \pid b, \pid s\}}{\mathit{ok}}\thru\med{a2cbs};
\\	4.\quad	& \phantom{\m{if} \ {\pid a.\mathit{happy}\ } \m{then}\ (\ }
			\{ \com{\pid a\langle\mathit{money}\rangle}{\pid b},
\com{\pid c\langle\mathit{book}\rangle}{\pid s} \}\thru\med{ac2bs};\nil)
\\	5.\quad	& \phantom{\m{if} \ {\pid a.\mathit{happy}}\ }
			\m{else}\ \left(\ \sel{\pid a}{\{ \pid c, \pid b, \pid s\}}{\mathit{ko}}\thru\med{a2cbs};\nil \right)
                        \hspace*{10em}\qed
\end{align*}
Our notation suggests that Alice multi-casts $\mathit{ok}$ (or $\mathit{ko}$) to Carol, the bank, and the shipper.
It is important to understand, though, that this notation is really just syntactic sugar.
Ultimately, only the behaviour of connector $\med{a2cbs}$ determines how exactly the three communications represented by
$\sel{\pid a}{\{ \pid c, \pid b, \pid s\}}{\ell}$ transpire.
\end{example}

The semantics of most terms is standard.
In $\gencond$, process $\pid p$ evaluates expression $e$; if this results in $\mathit{true}$, the choreography
proceeds as $C_1$, and otherwise, as $C_2$.
In $\genrec$, procedure $X$ is defined as $C_2$; it can then be invoked as $X$ from both $C_1$ and $C_2$.
$\nil$ indicates successful termination.

The interesting new part is the semantics of $\tilde\eta\thru\gamma$, which
informally specifies that all communications in $\tilde\eta$
occur through connector $\gamma$.
More precisely, in $\gencom\thru\gamma$, process $\pid p$ (the \emph{sender}) evaluates expression
$e$ and offers the resulting value to connector $\gamma$.
The connector eventually accepts and passes it to process $\pid q$ (the \emph{receiver}), who
stores it in its local variable $x$.
The behaviour of $\gensel\thru\gamma$ is similar, except that $\pid p$ offers a label instead of
a normal value.
Label selections do not change the state of the receiving process; their role in synthesizing
process implementations is discussed in \S~\ref{sec:processes}.

The boxed terminals in Fig.~\ref{fig:cc_syntax} are \emph{runtime terms}, meant to be used only in defining the formal semantics and not by programmers.
They arise because connectors may have a multi-step semantics (they do not necessarily synchronise sends with receives).
In particular, $\genrecv$ is obtained when a communication $\gencom$ is partially
executed, and $\pid p$ has already sent its value, but $\pid q$ has not yet received
it; $\genchoice$ arises similarly -- see Ex.~\ref{ex:motivating+reduc+inform} below.

For the semantics of $\tilde\eta\thru\gamma$ to be well-defined, $\tilde\eta$ must satisfy two conditions.
\label{cond:chor}
First, all interactions in $\tilde\eta$ must have \emph{distinct receivers}: if $\com{\pid p_1\carr{e_1}}{\pid q_1.x_1} \in \tilde\eta$ and $\com{\pid p_2\carr{e_2}}{\pid q_2.x_2} \in \tilde\eta$, then $\pid q_1 \neq \pid q_2$.
This ensures that the value received by a receiver is uniquely defined.
Second, all sends must be \emph{consistent}: if a process $\pid p$ is involved in
multiple interactions in the same set, then they are either all communications of the same expression or all selections of the same label.

\begin{example}[Book sale]
\label{ex:book-sale-conditions}
	In the context of our running example, the following interaction sets are allowed ($\yes$) or disallowed ($\no$) by our conditions for distinct receivers and consistent sends.

	\begin{itemize}
		\item[$\yes$]
		$\{ \com{\pid a\langle\mathit{money}\rangle}{\pid b}, \com{\pid c\langle\mathit{book}\rangle}{\pid s} \}$ -- Alice sends money to the bank, while Carol sends a book to the shipper (distinct receivers; consistent sends).

		\item[$\no$]
		$\{ \com{\pid a\langle\mathit{money}\rangle}{\pid b},
\com{\pid c\langle\mathit{money}\rangle}{\pid b} \}$
-- Both Alice and Carol send money to the bank, but the
bank can receive only from one sender at a time (receivers are not distinct).

		\item[$\no$]
		$\{ \com{\pid c\langle\mathit{price}\rangle}{\pid a},
\com{\pid c\langle\mathit{book}\rangle}{\pid s} \}$ -- Carol sends both the price to Alice and the
book to the shipper, but Carol can send only one value at a time (sends are not consistent).

		\item[$\yes$]
		$\{ \com{\pid a}{\pid b[\mathit{ok}]}, \com{\pid a}{\pid c[\mathit{ok}]}, \com{\pid a}{\pid s[\mathit{ok}]} \}$ -- Alice sends label $\mathit{ok}$ to Carol, the bank, and the shipper (distinct receivers; consistent sends).
		\qed
	\end{itemize}
\end{example}

\begin{example}[Book sale]\label{ex:motivating+reduc+inform}
	Returning to our running example, suppose that Alice wants to buy the book titled Foo.

	Let $\conn$ denote the connector mapping in Ex.~\ref{ex:motivating+conn}, let $C'$
        denote lines 2--3 in Ex.~\ref{ex:motivating+cr}, let $\sigma_0$ denote the initial
        choreography state function such that $\sigma_0(\pid a.\mathit{title}) = \code{"foo"}$, let
        $\sigma_0' = \sigma_0 [\pid c.\mathit{title} \mapsto \code{"foo"}]$, and let
        $\astate_0 = \{\med{a2c} \mapsto \langle 1, \{m \mapsto \bot\} \rangle \} \cup \astate_0^\text{rest}$
        denote the initial automaton state function,
	where:
	\begin{equation*}
		\astate_0^\text{rest} = \{
			\med{c2a} \mapsto \langle 1, \{m \mapsto \bot\} \rangle,
			\med{a2cbs} \mapsto \langle 1, \emptyset \rangle,
			\med{ac2bs} \mapsto \langle 1, \emptyset \rangle
		\}
	\end{equation*}
	Initially, thus, all connectors are in their initial state (state $1$).
	Furthermore, connectors $\med{a2c}$ and $\med{c2a}$ have an empty memory cell ($m \mapsto \bot$); connectors $\med{a2cbs}$ and $\med{ac2bs}$ have no memory cells (memory snapshot $\emptyset$).
	By rule \rname{C}{Com} (presented after this example), the choreography in Ex.~\ref{ex:motivating+cr} reduces under $\conn$ as follows:
	\setlength{\belowdisplayskip}{0pt}
	\begin{alignat*}{4}
		{} &
	&	\quad \com{\pid a\langle\mathit{title}\rangle}{\pid c}\thru\med{a2c};C', &
	&	&\quad \sigma_0,
	&	&\quad \{ \med{a2c} \mapsto \langle 1, \{m \mapsto \bot\} \rangle \} \cup \astate_0^\text{rest}
	\\	{\rightsquigarrow_{\conn}} \text{\rlap{\quad \fbox{\begin{minipage}[t]{.85\linewidth}
			\footnotesize
			In state $1$, according to $\conn$, connector $\med{a2c}$ only has a
                     transition that allows Alice to send (asynchronously) to Carol.
			By performing such a send, Alice enables the choreography to make a reduction, in which the first half of the communication completes; the (asynchronous) receive remains.
			In the same step, $\med{a2c}$ moves to state $2$, and the value sent by Alice is stored in $\med{a2c}$'s internal memory cell ($\mathit{title}$ evaluates to $\code{"foo"}$, based on $\sigma_0$).
		\end{minipage}}}}
	\\	{} &
	&	\quad \pid c.\mathit{title}?\code{"foo"}\thru\med{a2c};C', &
	&	&\quad \sigma_0,
	&	&\quad \{ \med{a2c} \mapsto \langle 2, \{m \mapsto \code{"foo"}\} \rangle \} \cup \astate_0^\text{rest}
	\\	{\rightsquigarrow_{\conn}} \text{\rlap{\quad\fbox{\begin{minipage}[t]{.85\linewidth}
			\footnotesize
			In state $2$, according to $\conn$, connector $\med{a2c}$ only has a
                        transition that allows Carol to receive (asynchronously) from Alice.
			By performing such a receive, Carol enables the choreography to make a reduction in which the whole communication completes.
			In the same step, $\med{a2c}$ moves to state $1$ (the internal memory cell is not cleared).
		\end{minipage}}}}
	\\	{} &
	&	\quad \emptyset\thru\med{a2c};C', &
	&	&\quad \sigma_0',
	&	&\quad \{ \med{a2c} \mapsto \langle 1, \{m \mapsto \code{"foo"} \} \rangle \} \cup \astate_0^\text{rest}
	\end{alignat*}
	\qed
\end{example}

\subsection{Formal semantics}
The semantics of CR is a reduction semantics parametrised over a \emph{connector mapping}: a function
$\conn$ from connector names to automata.
Intuitively, $\conn(\gamma)$ denotes the automaton that models connector $\gamma$ used in the
choreography; the set $P$ of ports in each $\conn(\gamma)$ is simply a one-to-one mapping (i.e., a renaming of) the set of processes that use the connector.\footnote{This means that each process can interact at
most through one port in each automaton.}

\begin{example}[Book sale]
\label{ex:motivating+conn}
	The connector names in Ex.~\ref{ex:motivating+cr} are $\med{a2c}$, $\med{c2a}$, $\med{a2cbs}$, and $\med{ac2bs}$.
	Thus, the requirements in Ex.~\ref{ex:motivating+reqs} give rise to the following connector mapping:
	\begin{equation*}
		\left\{ \!\begin{aligned}
			\med{a2c} & \mapsto \reo{Async1}[\pid a / p_1, \pid c / p_2] , &
			\med{c2a} & \mapsto \reo{Async1}[\pid c / p_1, \pid a / p_2] ,
		\\	\med{a2cbs} & \mapsto \reo{SyncMulti3}[\pid a / p_1, \pid b / p_2, \pid c / p_3, \pid s / p_4] , &
		\med{ac2bs} & \mapsto \reo{Barrier}[\pid a / p_1, \pid b / p_2, \pid c / p_3, \pid s / p_4]
		\end{aligned} \right\}
	\end{equation*}
	where $\reo{Async1} [\pid a / p_1, \pid c / p_2]$ denotes automaton $\reo{Async1}$ in
Fig.~\ref{fig:aut}, with $\pid a$ substituted for $p_1$, and $\pid c$ for $p_2$ (and likewise in the other mappings).
	Under this connector mapping, thus, Alice and Carol communicate via asynchronous channels ($\med{a2c}$ and $\med{c2a}$) to exchange title and price; Alice, Carol, the bank, and the shipper communicate via synchronous multi-cast ($\med{a2cbs}$) to exchange $\mathit{ok}$ or $\mathit{ko}$, and via barrier sends\slash receives ($\med{ac2bs}$) to exchange money and book.
	\qed
\end{example}

\begin{remark}
	The book sale scenario illustrates an important design decision, namely
	the separation between \emph{intention} and \emph{realisation}: a choreography
	defines \emph{what} interactions are intended (e.g., communications of the
	money from Alice to the Bank and the book from Carol to the shipper), while the
	connectors define \emph{how} these communications are realised (e.g.,
	synchronously or asynchronously). As a result, every interaction has to be
	expressed in two places, serving two complementary purposes: as
	``specifications'' in the choreography and as ``implementations'' in the
	connectors (automaton transitions). As usual, implementations should respect
	specifications; we address this in \S~\ref{sec:chor:deadlock}.
	\qed
\end{remark}

The reduction relation for CR under a given $\mathcal G$ is denoted as $\redM$; it ranges
over triples $C,\sigma,\astate$, where $C$ is a choreography, $\sigma$ is a choreography state function
(mapping each process to a mapping of its variables to values, i.e., $\sigma(\pid p.x)$ is the value stored at variable
$x$ in process $\pid p$), and $\astate$ is an automaton state function (mapping each connector name $\gamma$
in the domain of $\conn$ to a pair $\langle s,\mu\rangle$ of the state and memory snapshot of the
automaton~$\conn(\gamma)$).
Before introducing the formal rule for communications, we give an example that discusses the intuition.

\begin{example}[Book sale]\label{ex:motivating+reduc+inform}
	Returning to our running example, suppose that Alice wants to buy the book titled Foo.

	Let $\conn$ denote the connector mapping in Ex.~\ref{ex:motivating+conn}, let $C'$
        denote lines 2--3 in Ex.~\ref{ex:motivating+cr}, let $\sigma_0$ denote the initial
        choreography state function such that $\sigma_0(\pid a.\mathit{title}) = \code{"foo"}$, let
        $\sigma_0' = \sigma_0 [\pid c.\mathit{title} \mapsto \code{"foo"}]$, and let
        $\astate_0 = \{\med{a2c} \mapsto \langle 1, \{m \mapsto \bot\} \rangle \} \cup \astate_0^\text{rest}$
        denote the initial automaton state function,
	where:
	\begin{equation*}
		\astate_0^\text{rest} = \{
			\med{c2a} \mapsto \langle 1, \{m \mapsto \bot\} \rangle,
			\med{a2cbs} \mapsto \langle 1, \emptyset \rangle,
			\med{ac2bs} \mapsto \langle 1, \emptyset \rangle
		\}
	\end{equation*}
	Initially, thus, all connectors are in their initial state (state $1$).
	Furthermore, connectors $\med{a2c}$ and $\med{c2a}$ have an empty memory cell ($m \mapsto \bot$); connectors $\med{a2cbs}$ and $\med{ac2bs}$ have no memory cells (memory snapshot $\emptyset$).
	By rule \rname{C}{Com} (presented after this example), the choreography in Ex.~\ref{ex:motivating+cr} reduces under $\conn$ as follows:
	\setlength{\belowdisplayskip}{0pt}
	\begin{alignat*}{4}
		{} &
	&	\quad \com{\pid a\langle\mathit{title}\rangle}{\pid c}\thru\med{a2c};C', &
	&	&\quad \sigma_0,
	&	&\quad \{ \med{a2c} \mapsto \langle 1, \{m \mapsto \bot\} \rangle \} \cup \astate_0^\text{rest}
	\\	{\rightsquigarrow_{\conn}} \text{\rlap{\quad \fbox{\begin{minipage}[t]{.85\linewidth}
			\footnotesize
			In state $1$, according to $\conn$, connector $\med{a2c}$ only has a
                     transition that allows Alice to send (asynchronously) to Carol.
			By performing such a send, Alice enables the choreography to make a reduction, in which the first half of the communication completes; the (asynchronous) receive remains.
			In the same step, $\med{a2c}$ moves to state $2$, and the value sent by Alice is stored in $\med{a2c}$'s internal memory cell ($\mathit{title}$ evaluates to $\code{"foo"}$, based on $\sigma_0$).
		\end{minipage}}}}
	\\	{} &
	&	\quad \pid c.\mathit{title}?\code{"foo"}\thru\med{a2c};C', &
	&	&\quad \sigma_0,
	&	&\quad \{ \med{a2c} \mapsto \langle 2, \{m \mapsto \code{"foo"}\} \rangle \} \cup \astate_0^\text{rest}
	\\	{\rightsquigarrow_{\conn}} \text{\rlap{\quad\fbox{\begin{minipage}[t]{.85\linewidth}
			\footnotesize
			In state $2$, according to $\conn$, connector $\med{a2c}$ only has a
                        transition that allows Carol to receive (asynchronously) from Alice.
			By performing such a receive, Carol enables the choreography to make a reduction in which the whole communication completes.
			In the same step, $\med{a2c}$ moves to state $1$ (the internal memory cell is not cleared).
		\end{minipage}}}}
	\\	{} &
	&	\quad \emptyset\thru\med{a2c};C', &
	&	&\quad \sigma_0',
	&	&\quad \{ \med{a2c} \mapsto \langle 1, \{m \mapsto \code{"foo"} \} \rangle \} \cup \astate_0^\text{rest}
	\end{alignat*}
	\qed
\end{example}

These intuitions are captured in the rule for communications \rname{C}{Com}, which is the key rule defining $\redM$.
We discuss this rule in detail.

\[
\infer[\rname{C}{Com}]
  {\tilde\eta\thru\gamma;C,\sigma,\astate \redM \tilde\eta'\thru\gamma;C,\sigma',\astate[\gamma\mapsto\langle s',\mu'\rangle]}
  {
    \astate(\gamma)=\langle s,\mu\rangle
    &
    \tilde\eta,\sigma,\mu \tr[\phi] \tilde\eta',\sigma',\mu'
    &
    s \tr[\phi]_\gamma s'
 }
\]

Rule $\rname{C}{Com}$ allows some of the communications in $\tilde\eta$ to
reduce as long as the state of the automaton corresponding
to connector $\gamma$ allows it.
The rule reads: under $\conn$, triple $\tilde\eta\thru\gamma;C, \sigma, \astate$ can reduce if automaton
$\conn(\gamma)$ can fire a transition out of its current state that is compatible with the interactions specified in
$\tilde\eta$.

More formally, the first premise of this rule retrieves the current state $s$ and memory snapshot $\mu$ of the connector $\gamma$
controlling the communication. In the second premise, the labelled reduction
$\tilde\eta,\sigma,\mu\tr[\phi]\tilde\eta',\sigma',\mu'$ (defined below) states that reducing $\til\eta$ to $\til\eta'$ transforms
the state of processes $\sigma$ into $\sigma'$ and the memory snapshot of the connector $\mu$ into $\mu'$.
The label $\phi$ represents the actions executed in this reduction. The third premise checks that these actions are
allowed by the automaton, by checking that $\phi$ labels an outgoing transition of $s$.

\begin{figure}[t]
\small
\begin{eqnarray*}
  &
  \infer[\rname{C}{SendVal}]
        {\{\gencom\},\sigma,\mu\tr[\pid p\flow m]\genrecv,\sigma,\mu[m\mapsto v]\rule[1em]{0em}{0em}}
        {e\eval\sigma{\pid p}v}
  \quad
  \infer[\rname{C}{Mem}]
        {\emptyset,\sigma,\mu\tr[m_1\flow m_2]\emptyset,
        \sigma,\mu[m_2 \mapsto v]\rule[1em]{0em}{0em}}
        {\mu(m_1)=v}
  \\[1ex]
  &
  \infer[\rname{C}{RecvVal}]
        {\{\genrecv\},\sigma,\mu\tr[m\flow\pid q]\emptyset,\sigma[\pid q.x\mapsto v],\mu\rule[1em]{0em}{0em}}
        {\mu(m)=v}
  \quad
  \infer[\rname{C}{Mon}]
        {(\tilde\eta_1\uplus\tilde\eta'),\sigma,\mu\tr[\phi](\tilde\eta_2\uplus\tilde\eta'),\sigma',\mu'\rule[1em]{0em}{0em}}
        {\tilde\eta_1,\sigma,\mu\tr[\phi]\tilde\eta_2,\sigma',\mu'}
  \\[1ex]
	&
	\infer[\rname{C}{SyncVal}]
        {\{\gencom\},\sigma,\mu\tr[\pid p\flow\pid q]\emptyset,\sigma[\pid q.x\mapsto v],\mu\rule[1em]{0em}{0em}}
        {e\eval\sigma{\pid p}v}
		\quad
	\infer[\rname{C}{SyncSel}]
        {\{\gensel\},\sigma,\mu\tr[\pid p\flow\pid q]\emptyset,\sigma,\mu\rule[1em]{0em}{0em}}
        {}
  \\[1ex]
  &
  \infer[\rname{C}{SendSel}]
        {\{\gensel\},\sigma,\mu\tr[\pid p\flow m]\{\genchoice\},\sigma,\mu[m\mapsto\ell]\rule[1em]{0em}{0em}}
        {}
	\quad
	\infer[\rname{C}{RecvSel}]
        {\{\genchoice\},\sigma,\mu\tr[m\flow\pid q]\emptyset,\sigma,\mu\rule[1em]{0em}{0em}}
        {\mu(m)=\ell}
  \\[1ex]
  &
  \infer[\rname{C}{Join}]
        {(\tilde\eta_1\uplus\tilde\eta_2),\sigma,\mu\tr[\phi_1\cup\phi_2](\tilde\eta_1'\uplus\tilde\eta_2'),\sigma'',\mu''\rule[1em]{0em}{0em}}
        {
          \tilde\eta_1,\sigma,\mu\tr[\phi_1]\tilde\eta_1',\sigma',\mu'
          &
          \tilde\eta_2,\sigma',\mu'\tr[\phi_2]\tilde\eta_2',\sigma'',\mu''
          &
          (\dagger)
        }
\end{eqnarray*}

\caption{Semantics of individual communications. The side condition $(\dagger)$ in rule
\rname{C}{Join} reads: if memory cell $m$ occurs in both $\phi_1$ and $\phi_2$, then it is not
both written to in $\phi_1$ and read from in $\phi_2$.}
\label{fig:compat}
\end{figure}
The rules defining labelled reductions $\tilde\eta,\sigma,\mu\tr[\phi]\tilde\eta',\sigma',\mu'$ are given in
Fig.~\ref{fig:compat}.
They are obtained by considering the different possible cases for terms $\til\eta$ and matching them to appropriate constraints
$\phi$.
Let $e\eval\sigma{\pid p}v$ denote that expression $e$ evaluates to value $v$ under $\sigma$ (i.e., after substituting every free variable $x$ in $e$ by $\sigma(\pid p.x)$).
In rule \rname{C}{SyncVal}, an entire communication $\gencom$ is executed in one step.
Accordingly, the label $\pid p\flow \pid q$ denotes that the automaton should support a synchronous communication between
$\pid p$ and $\pid q$.
The state of the receiver $\pid q$ is updated with the value sent by $\pid p$.
Rule~\rname{C}{SendVal} applies in the case where the message from $\pid p$ should be stored in a memory cell $m$ of the automaton
(label $\pid p\flow m$).
This is used for asynchronous communications, where the message is received later on by the receiver.
In the reductum, we keep a runtime term signalling that the receiver is still waiting to receive the message ($\genrecv$).
This kind of runtime terms is handled by \rname{C}{RecvVal}, whose label specifies that $\pid q$ should receive the
message stored in some memory cell $m$.
Its premise checks that the value $\pid q$ is expecting to receive ($v$, defined in the choreography term)
is the one stored in $m$.%
\footnote{In other words, the automaton delivers the messages as specified in the choreography.}
Rules \rname{C}{SyncSel}, \rname{C}{SendSel} and \rname{C}{RecvSel} deal with
label selections similarly.
Rule \rname{C}{Mem} covers internal transitions in the automaton that only modify memory.
Finally, rules \rname{C}{Mon} and \rname{C}{Join} extend this notion to transitions labelled by non-singleton sets. Rule
\rname{C}{Mon} states that some communications in $\til\eta$ may not be executed at all (i.e., they are postponed until a later reduction). Rule \rname{C}{Join} allows
executing several communications at the same time; note that labels $\phi_1$ and $\phi_2$ may share constraints
(e.g., with multi-cast), while $\til\eta_1$ and $\til\eta_2$ must be disjoint ($\uplus$ is the
disjoint union operator).

The syntactic assumptions on choreography terms (page \pageref{cond:chor}), namely distinct receivers (all terms on the right-hand sides of interactions are
distinct) and consistent sends (if a process sends to several processes, then it always sends
the same value or label), ensure that the sequentialisation in \rname{C}{Join} is of no
consequence (building the final set $\phi$ in any order always yields the same final states
$\sigma'$ and $\mu'$ in \rname{C}{Com}), except if the same memory cell is both written to and read from.
In that case, the read must precede the write; this is guaranteed by side condition $(\dagger)$, which ensures that concurrent accesses to a memory cell are done in a consistent way.

\begin{example}[Book sale]\label{ex:motivating+reduc+a2c}
	We formally derive the reductions in Ex.~\ref{ex:motivating+reduc+inform}, using rule \rname{C}{Com} and Fig.~\ref{fig:compat}.
	Let $\conn$, $\sigma_0$, $\sigma_0'$, and $\astate_0$ be defined as in Ex.~\ref{ex:motivating+reduc+inform}.

	For the first reduction, first, $\astate_0(\med{a2c}) = \langle 1, \{m \mapsto \bot\} \rangle$ (rule \rname{C}{Com}, first premise).
	Next, automaton $\conn(\med{a2c})$ has one transition out of state $1$, namely $1 \tr[\pid a \flow m]_{\med{a2c}} 2$ (rule \rname{C}{Com}, third premise).
	Finally, using Fig.~\ref{fig:compat}, we need to derive $\com{\pid a\langle\mathit{title}\rangle}{\pid c}, \sigma_0, \{m \mapsto \bot\} \tr[\pid a \flow m] \ldots$ (rule \rname{C}{Com}, second premise); this follows from rule \rname{C}{SendVal}.
	Thus, we derive:
	\begin{equation*}
		\dfrac{
			\raisebox{-7\jot}{$\astate_0(\med{a2c}) = \langle 1, \{m \mapsto \bot\} \rangle$} \qquad \dfrac{
				\mathit{title}\eval{\sigma_0}{\pid a} \code{"foo"}
			}{\begin{array}{@{}c@{}}
				\com{\pid a\langle\mathit{title}\rangle}{\pid c}, \sigma_0, \{m \mapsto \bot\} \tr[\pid a \flow m] {}
			\\	\pid c.\mathit{title}?\code{"foo"}, \sigma_0, \{m \mapsto \code{"foo"}\}
			\end{array}} \raisebox{.5pt}{\rlap{\scriptsize\rname{C}{SendVal}}} \qquad \raisebox{-7\jot}{$1 \tr[\pid a \flow m]_{\med{a2c}} 2$}
		}{
			\begin{array}{@{} r @{\enspace} r @{\enspace} l @{\enspace} r @{}}
				& \com{\pid a\langle\mathit{title}\rangle}{\pid c}\thru\med{a2c};C', & \sigma_0, & \{ \med{a2c} \mapsto \langle 1, \{m \mapsto \bot\} \rangle \} \cup \astate_0^\text{rest}
			\\	\rightsquigarrow_\conn & \pid c.\mathit{title}?\code{"foo"}\thru\med{a2c};C', & \sigma_0, & \{ \med{a2c} \mapsto \langle 2, \{m \mapsto \code{"foo"}\} \rangle \} \cup \astate_0^\text{rest}
			\end{array}
		}
	\end{equation*}
	Henceforth, let $\astate_0' = \{ \med{a2c} \mapsto \langle 2, \{m \mapsto \code{"foo"}\} \rangle \} \cup \astate_0^\text{rest}$.

	For the second reduction, similarly, we derive:
	\setlength{\belowdisplayskip}{0pt}
	\begin{equation*}
		\dfrac{
			\raisebox{-7\jot}{$\astate_0'(\med{a2c}) = \langle 2, \{m \mapsto \code{"foo"}\} \rangle$} \qquad \dfrac{
				\{m \mapsto \code{"foo"}\}(m) = \code{"foo"}
			}{\begin{array}{@{}c@{}}
				\pid c.\mathit{title}?\code{"foo"}, \sigma_0, \{m \mapsto \code{"foo"}\}
			\\	{} \tr[m \flow \pid c] \emptyset, \sigma_0', \{ m \mapsto \code{"foo"} \}
			\end{array}} \raisebox{.5pt}{\rlap{\scriptsize\rname{C}{RecvVal}}} \qquad \raisebox{-7\jot}{$2 \tr[m \flow \pid c]_{\med{a2c}} 1$}
		}{
			\begin{array}{@{} r @{\enspace} r @{\enspace} l @{\enspace} r @{}}
				& \pid c.\mathit{title}?\code{"foo"}\thru\med{a2c};C', & \sigma_0, & \{ \med{a2c} \mapsto \langle 2, \{m \mapsto \code{"foo"}\} \rangle \} \cup \astate_0^\text{rest}
			\\	\rightsquigarrow_\conn & \emptyset\thru\med{a2c};C', & \sigma_0', & \{ \med{a2c} \mapsto \langle 1, \{m \mapsto \code{"foo"}\} \rangle \} \cup \astate_0^\text{rest}
			\end{array}
		}
	\end{equation*}
	\qed
\end{example}

\begin{example}[Book sale]\label{ex:motivating+reduc+ac2bs}
	We formally derive the reduction of line 4 in Ex.~\ref{ex:motivating+cr}.
	Let $\conn$ denote the connector mapping in Ex.~\ref{ex:motivating+conn}, let $\sigma$ denote a choreography state function such that $\sigma(\pid a.\mathit{money}) = \code{\$10}$ and $\sigma(\pid c.\mathit{book}) = \code{foo.pdf}$, let $\sigma' = \sigma [\pid b.\mathit{money} \mapsto \code{\$10}]$, let $\sigma'' = \sigma' [\pid s.\mathit{book} \mapsto \code{foo.pdf}]$, and let $\astate = \astate^\text{rest} \cup \{ \med{ac2bs} \mapsto \{1, \emptyset\} \}$ denote an automaton state function, for some $\astate^\text{rest}$.
	We derive:
	\setlength{\belowdisplayskip}{0pt}
	\begin{equation*}
		\dfrac{
			\raisebox{-7\jot}{$\astate(\med{ac2bs}) = \langle 1, \emptyset \rangle$} \quad \dfrac{
				\dfrac{
					\mathit{money}\eval\sigma{\pid a}\code{\$10}
				}{
					\begin{array}{@{}c@{}}
						\{ \com{\pid a\langle\mathit{money}\rangle}{\pid b}\}, \sigma, \emptyset
					\\	{} \tr[\pid a \flow \pid b] \emptyset, \sigma', \emptyset
					\end{array}
				} \qquad \dfrac{
					\mathit{book}\eval\sigma{\pid c}\code{foo.pdf}
				}{
					\begin{array}{@{}c@{}}
						\{ \com{\pid c\langle\mathit{book}\rangle}{\pid s} \}, \sigma', \emptyset
					\\	{} \tr[\pid c \flow \pid s] \emptyset, \sigma'', \emptyset
					\end{array}
				} \raisebox{.5pt}{\rlap{\scriptsize\rname{C}{SyncVal}}}
			}{
				\begin{array}{@{}c@{}}
					\{ \com{\pid a\langle\mathit{money}\rangle}{\pid b}, \com{\pid c\langle\mathit{book}\rangle}{\pid s} \}, \sigma, \emptyset
				\\	{} \tr[\pid a \flow \pid b \wedge \pid c \flow \pid s] \emptyset, \sigma'', \emptyset
				\end{array}
			} \raisebox{.5pt}{\rlap{\scriptsize\rname{C}{Join}}} \qquad \raisebox{-7\jot}{$1 \tr[\pid a \flow \pid b \wedge \pid c \flow \pid s]_{\med{ac2bs}} 1$}
		}{
			\begin{array}{@{}r@{\enspace}r@{\enspace}l@{\enspace}r@{}}
				& \{ \com{\pid a\langle\mathit{money}\rangle}{\pid b}, \com{\pid c\langle\mathit{book}\rangle}{\pid s} \}\thru\med{ac2bs};\nil, & \sigma, & \astate
			\\	\rightsquigarrow_\conn & \emptyset\thru\med{ac2bs};\nil, & \sigma'', & \astate
			\end{array}
		}
	\end{equation*}
	\qed
\end{example}

Rule \rname{C}{Com} is the only rule in the semantics that can cause a choreography to get stuck; we
discuss this in more detail in the next section.

\begin{figure}[t]
\small
\begin{eqnarray*}
&
\infer[\rname{C}{Eta-Eta}]
{
	\left( \tilde\eta\thru\gamma;\wtil{\eta'}\thru\gamma' \right) \equiv \left( \wtil{\eta'}\thru\gamma';\tilde\eta\thru\gamma \right)
}
{
	\pn(\tilde\eta) \cap \pn(\wtil{\eta'}) = \emptyset
        &
        \gamma\neq\gamma'
}
\\[1ex]
&
\infer[\rname{C}{Eta-Split}]
{
  \left( \wtil{\eta_1}\thru\gamma;\wtil{\eta_2}\thru\gamma \right) \equiv \left(\wtil{\eta_1}\cup\wtil{\eta_2}\right)\thru\gamma
}
{
  \pn(\wtil{\eta_1})\cap\pn(\wtil{\eta_2}) = \emptyset
}
\\[1ex]
&
\infer[\rname{C}{Eta-Cond}]
{
	\left( \cond{\pid p.e}{(\tilde\eta \thru \gamma;C_1)}{(\tilde\eta \thru \gamma;C_2)} \right)
	\equiv
	\left( \tilde\eta \thru \gamma;\gencond \right)
}
{
	\pid p \not\in \pn(\tilde\eta)
}
\\[1ex]
&
\infer[\rname{C}{Eta-Rec}]
{
	\left( \rec{X}{C_2}{(\tilde\eta \thru \gamma;C_1)} \right)
	\equiv
	\left( \tilde\eta \thru \gamma;\genrec \right)
}
{
	\pn(C_i) \cap \pn(\tilde\eta) = \emptyset
}
\\[1ex]
&\infer[\rname{C}{Cond-Cond}]
{
	\begin{array}{c}
	\cond{\pid p.e}{
 		\left(\cond{\pid q.e'}
 		{C_1}{C_2}\right)
	}{
 		\left(\cond{\pid q.e'}
 		{C'_1}{C'_2}\right)
	}
	\\
	\equiv
	\\
	\cond{\pid q.e'}{
		\left(\cond{\pid p.e}
		{C_1}{C'_1}\right)
	}{
		\left(\cond{\pid p.e}
		{C_2}{C'_2}\right)
	}
	\end{array}
}
{
	\pid p \neq \pid q 
}
\\[1ex]
&\infer[\rname{C}{Unfold}]
{
	\left( \rec{X}{C_2}{C_1[\gencall]} \right)
	\precongr
	\left( \rec{X}{C_2}{C_1[C_2]} \right)
}
{}
\\[1ex]
&
\infer[\rname{C}{EtaEnd}]
{
  \emptyset\thru\gamma;C \precongr C
}
{}
\qquad
\infer[\rname{C}{ProcEnd}]
{
	\left(\rec{X}{C}{\nil} \right) \precongr \nil
}
{}
\end{eqnarray*}
\caption{Cho-Reo-graphy, structural precongruence.}
\label{fig:mcr_precongr}
\end{figure}

The remaining rules defining $\redM$ are standard from other choreography calculi (e.g.,~\cite{CM16b}), and they are given Fig.~\ref{fig:mcr_semantics}):
they define reductions for conditionals, allow for reductions under procedure
definitions, and close the reduction relation $\redM$ under the structural
precongruence $\precongr$.
\begin{figure}[t]
\begin{eqnarray*}
&\infer[\rname{C}{Cond}]
{
	\gencond, \sigma, \astate \redM C_i, \sigma, \astate
}
{
	i = 1 \ \text{if } e\eval\sigma{\pid p}\mathit{true},\ 
	i = 2 \ \text{otherwise}
}
\\[1ex]
&\infer[\rname{C}{Struct}]
{
	C_1, \sigma, \astate \redM C'_1, \sigma', \astate'
}
{
	C_1 \precongr C_2
	& C_2, \sigma, \astate \redM C'_2, \sigma', \astate'
	& C'_2 \precongr C'_1
}
\\[1ex]
&
\infer[\rname{C}{Ctx}]
{
	\genrec, \sigma, \astate \redM
	\rec{X}{C_2}{C'_1}, \sigma', \astate'
}
{
	C_1, \sigma, \astate \redM C'_1, \sigma', \astate
}
\end{eqnarray*}
\caption{Cho-Reo-graphy, semantics.}
\label{fig:mcr_semantics}
\end{figure}
Rule~\rname{C}{Struct} uses a structural precongruence that allows for actions to be swapped if they do not interfere.
The most interesting rules deal with communications: \rname{C}{Eta-Split} and
\rname{C}{EtaEnd}, displayed in Fig.~\ref{fig:mcr_precongr} (where
$\pn(\tilde\eta)$ denotes the process names that occur in~$\tilde\eta$).
Rule~\rname{C}{Eta-Split} allows interactions through the same connector to be joined in one $\tilde \eta$ or split among
several ones, which is necessary for correct interaction with rule \rname{C}{Eta-Eta}.
(See~\cite{lcf:ksl:fm:17} for a similar discussion: this rule is needed whenever one $\tilde\eta$ can specify several communications. Ex.~\ref{ex:ksl17} below illustrates this interplay in our language.)
Rule \rname{C}{EtaEnd} removes completed interactions from the head of a choreography.
The remaining rules defining $\precongr$ concern recursion unfolding and swapping of conditionals; these rules are straightforward adaptations of those in \cite{CM16b}.

\begin{example}
\label{ex:ksl17}
  Consider the following choreography, where for simplicity we abstract from the actual values being communicated.
  \[C \equiv \{\com{\pid p}{\pid q},\com{\pid r}{\pid s}\}\thru\gamma;\{\com{\pid p}{\pid q},\com{\pid t}{\pid v}\}\thru\gamma'\]
  and assume that both $\conn(\gamma)$ and $\conn(\gamma')$ allow the interactions between the processes they connect to occur
  independently and in any order.
  Then it is actually possible that the communication between $\pid t$ and $\pid v$ is the first one to take place.
  In order for our choreography language to allow this behaviour, we need to use both \rname{C}{Eta-Split} and
  \rname{C}{Eta-Eta} to exchange actions, as follows.
  \setlength{\belowdisplayskip}{0pt}
  \begin{align*}
    C
    & \equiv
    \{\com{\pid p}{\pid q},\com{\pid r}{\pid s}\}\thru\gamma;\{\com{\pid p}{\pid q},\com{\pid t}{\pid v}\}\thru\gamma' \\
    & \equiv
    \{\com{\pid p}{\pid q},\com{\pid r}{\pid s}\}\thru\gamma;\com{\pid t}{\pid v}\thru\gamma';\com{\pid p}{\pid q}\thru\gamma'
    & \mbox{by} & \rname{C}{Eta-Split}\\
    & \equiv
    \com{\pid t}{\pid v}\thru\gamma';\{\com{\pid p}{\pid q},\com{\pid r}{\pid s}\}\thru\gamma;\com{\pid p}{\pid q}\thru\gamma'
    & \mbox{by} & \rname{C}{Eta-Eta}
  \end{align*}
  \qed
\end{example}

\subsection{Flexibility}

An immediate advantage of CR is that different communication semantics can freely be mixed in the same
choreography. A second advantage is that CR enables programmers to change the semantics of a choreography
modularly, by altering the behaviour of the connectors through which the processes interact with each other,
\emph{without} the need to change the choreography itself.

\begin{example}[Book sale]\label{ex:motivating+flex}
	The original book sale scenario (Ex.~\ref{ex:motivating} and~\ref{ex:motivating+reqs})
        requires Alice and Carol to send money and book to the bank and the shipper synchronously,
        as they initially do not trust each other.
	Now, suppose Alice and Carol establish mutual trust after successfully completing a number of book sales, such that their communications with the bank and the shipper no longer need to occur synchronously.
	Instead of redeveloping the choreography from scratch, we need to redefine only the connector mapping $\conn$ in Ex.~\ref{ex:motivating+conn}, as follows:
	\begin{equation*}
		\conn := \conn [\med{ac2bs} \mapsto \raisebox{.67ex}{\begin{tikzautom}
			\state{S1}{0,0}{\co2}
			\state{S2}{right of=S1}{1}
			\state{S3}{right of=S2}{2}
			\trans{S1}{S2}{from16,to9}{below}{\pid a \flow \pid b}
			\trans{S2}{S1}{from8,to1}{above}{\pid c \flow \pid s}
			\trans{S2}{S3}{from16,to9}{below}{\pid a \flow \pid b}
			\init{below of}{S2}
			\trans{S3}{S2}{from8,to1}{above}{\pid c \flow \pid s}
		\end{tikzautom}}]
	\end{equation*}
	Thus, we updated the mapping for $\med{ac2bs}$; for all other connectors, the mapping remains
the same as in Ex.~\ref{ex:motivating+conn}.
	The new automaton for $\med{ac2bs}$ allows either a communication between Alice and the bank, asynchronously followed by a communication between Carol and the shipper (via state $2$), or the same two communications in the reverse order (via state $\co2$).\footnote{%
		The communications between Alice and the bank, and between Carol and the shipper, are synchronous in this automaton.
		We can easily make those communications asynchronous as well, but we skip this here to save space (the automaton gets larger).
	}

	Redefining the connector mapping for $\med{ac2bs}$ is the only change we need to make: the choreography itself is \emph{exactly} the same as in Ex.~\ref{ex:motivating+cr}.
	This means that also the first reductions remain \emph{exactly} the same as in Ex.~\ref{ex:motivating+reduc+inform} and~\ref{ex:motivating+reduc+a2c}.
	By contrast, the reduction in Ex.~\ref{ex:motivating+reduc+ac2bs} is no longer valid, as it
	relies on the semantics of $\med{ac2bs}$.
	To show the difference formally, let $\sigma$, $\sigma'$, $\sigma''$, and $\astate$ be defined as in Ex.~\ref{ex:motivating+reduc+ac2bs}.
	Let also $\astate' = \astate [\med{ac2bs} \mapsto \langle 2, \emptyset \rangle]$ and $\astate'' = \astate$.
	We derive:
	\begin{equation*}
		\dfrac{
			\raisebox{-7\jot}{$\astate(\med{ac2bs}) = \langle 1, \emptyset \rangle$} \qquad \dfrac{
				\dfrac{
					\mathit{money}\eval\sigma{\pid a}\code{\$10}
				}{
					\{ \com{\pid a\langle\mathit{money}\rangle}{\pid b}\}, \sigma, \emptyset \tr[\pid a \flow \pid b] \emptyset, \sigma', \emptyset
				} \raisebox{.5pt}{\rlap{\scriptsize\rname{C}{SyncVal}}}
			}{
				\begin{array}{@{}c@{}}
					\{ \com{\pid a\langle\mathit{money}\rangle}{\pid b}, \com{\pid c\langle\mathit{book}\rangle}{\pid s} \}, \sigma, \emptyset
				\\	{} \tr[\pid a \flow \pid b] \{ \com{\pid c\langle\mathit{book}\rangle}{\pid s} \}, \sigma', \emptyset
				\end{array}
			} \raisebox{.5pt}{\rlap{\scriptsize\rname{C}{Mon}}} \qquad \raisebox{-7\jot}{$1 \tr[\pid a \flow \pid b]_{\med{ac2bs}} 2$}
		}{
			\begin{array}{@{}r@{\enspace}r@{\enspace}l@{\enspace}l@{}}
				& \{ \com{\pid a\langle\mathit{money}\rangle}{\pid b}, \com{\pid c\langle\mathit{book}\rangle}{\pid s} \}\thru\med{ac2bs};\nil, & \sigma, & \astate
			\\	\rightsquigarrow_\conn & \{ \com{\pid c\langle\mathit{book}\rangle}{\pid s} \}\thru\med{ac2bs};\nil, & \sigma', & \astate'
			\end{array}
		}
	\end{equation*}
	Next, we derive:
	\begin{equation*}
		\dfrac{
			\raisebox{-2.25\jot}{$\astate(\med{ac2bs}) = \langle 2, \emptyset \rangle$} \quad \dfrac{
				\mathit{book}\eval\sigma{\pid c}\code{foo.pdf}
			}{
				\{ \com{\pid c\langle\mathit{book}\rangle}{\pid s} \}, \sigma', \emptyset \tr[\pid c \flow \pid s] \emptyset, \sigma'', \emptyset
			} \raisebox{.5pt}{{\scriptsize\rname{C}{SyncVal}}} \quad \raisebox{-2.25\jot}{$2 \tr[\pid c \flow \pid s]_{\med{ac2bs}} 1$}
		}{
			\begin{array}{@{}r@{\enspace}r@{\enspace}l@{\enspace}l@{}}
				& \{ \com{\pid c\langle\mathit{book}\rangle}{\pid s} \}\thru\med{ac2bs};\nil, & \sigma', & \astate'
			\\	\rightsquigarrow_\conn & \emptyset\thru\med{ac2bs};\nil, & \sigma'', & \astate''
			\end{array}
		}
	\end{equation*}
	Thus, as intended, our reduction rules let us derive two separate reductions with one communication each (first Alice and the bank, then Carol and the shipper) instead of one reduction with two communications (Ex.~\ref{ex:motivating+reduc+ac2bs}).
	Similarly, we can derive two separate reductions whereby Carol and the shipper communicate first, followed by Alice and the bank.
	\qed
\end{example}

\begin{example}
\label{ex:motivating+opposite}
	The previous example works also ``in the opposite direction'', from a
	trusting Alice and Carol (using connector mapping $\conn$ in
	Ex.~\ref{ex:motivating+flex}) to cautious ones (using connector mapping
	$\conn$ in Ex.~\ref{ex:motivating+conn}).

	Our choreography is also compatible with the case where we have a trusting Alice and
	a cautious Carol, who only sends the book after receiving payment.
	A connector mapping that implements this behaviour is the following.
	\begin{equation*}
		\conn := \conn [\med{ac2bs} \mapsto \raisebox{.67ex}{\begin{tikzautom}
			\state{S2}{right of=S1}{1}
			\state{S3}{right of=S2}{2}
			\trans{S2}{S3}{from16,to9}{below}{\pid a \flow \pid b}
			\init{below of}{S2}
			\trans{S3}{S2}{from8,to1}{above}{\pid c \flow \pid s}
		\end{tikzautom}}]
	\end{equation*}
	This connector mapping is still compatible with the choreography in
	Ex.~\ref{ex:motivating+cr}.
	The symmetric case where Carol is trusting and Alice is cautious is similar.
	\qed
\end{example}

\begin{example}
\label{ex:motivating+referee}
	Note that the changes to connector mappings in Ex.~\ref{ex:motivating+flex} and
	\ref{ex:motivating+opposite} would still be possible
	if the programmer had written, e.g.,
	\begin{equation*}
		...; \{ \com{\pid a\langle\mathit{money}\rangle}{\pid b} \}\thru\med{ac2bs};
		\{ \com{\pid c\langle\mathit{book}\rangle}{\pid s} \}\thru\med{ac2bs};\nil
	\end{equation*}
	instead of the choreography in Ex.~\ref{ex:motivating+cr}.
	Indeed, these two choreographies are equivalent due to the congruence rule
	$\rname{C}{Eta-Split}$, and thus the sets of connectors that are compatible with
	each of them are the same.

	This might be surprising at first, but it fits with the view of choreographies as
	global specifications of independent processes. Specifically in this case, no choreography
	can impose a causal dependency between the two communications
	$\com{\pid a\langle\mathit{money}\rangle}{\pid b}$ and
	$\com{\pid c\langle\mathit{book}\rangle}{\pid s}$ unless it includes an additional
	communication in the middle involving a process that can observe both.
	The lack of causal dependencies in this example thus leaves the connector for $\med{ac2bs}$
	free to decide the order in which the interactions are performed.
	\qed
\end{example}

\subsection{Deadlock-freedom}
\label{sec:chor:deadlock}

Rule \rname{C}{Com} is the only rule in the semantics that can cause a choreography to get stuck: in choreography $\tilde\eta\thru\gamma;C$, there can be
incompatibilities between the communications allowed by connector $\gamma$ and the intended communications in $\tilde\eta$, causing none of the communications in $\tilde\eta$ to be permitted by $\gamma$.
In this case, we say that $\gamma$ does not \emph{respect} the choreography.\footnote{A
choreography expresses the intentions of the programmer. Although she may instantiate
connectors however she likes, we assume they do not violate her intentions.}
Concretely, this can happen because the transitions available at the current state $s$ either require communications between
processes not involved in $\tilde\eta$ or because there is an incompatibility with messages in transit (the premises of rules
\rname{C}{RecvVal} and \rname{C}{RecvSel}).

In more detail, the first premise in rule \rname{C}{Com} always holds (assuming $\astate$ is defined for all connector names in $\tilde\eta\thru\gamma;C$).
This gives us unique bindings for $s$ and $\mu$.
The third premise in \rname{C}{Com} is also always true (assuming every state of a connector has at least one outgoing transition; this can trivially be checked).
For every outgoing transition of $s$, this gives us bindings for $\phi$ and $s'$.
Now, the choreography gets stuck if for each of those bindings, the second premise in \rname{C}{Com} is false.
This can happen in two cases: either $\tilde\eta,\sigma,\mu \tr[\varphi] \tilde\eta',\sigma',\mu'$ can be derived (using the rules in Fig.~\ref{fig:compat}) and $\phi \neq \varphi$ for every derivation, or $\tilde\eta,\sigma,\mu \tr[\varphi] \tilde\eta',\sigma',\mu'$ cannot be derived at all.
The former happens if every $\varphi$ contains different processes than $\phi$
(see Ex.~\ref{ex:motivating+resp1}), or the same processes but in different send\slash
receive pairs (see Ex.~\ref{ex:motivating+resp2}); the latter happens if $\tilde\eta$
contains only asynchronous receives for which rules \rname{C}{RecvVal} and \rname{C}{RecvSel} in
Fig.~\ref{fig:compat} are inapplicable (see Ex.~\ref{ex:motivating+resp3}).

\begin{example}[Book sale]\label{ex:motivating+resp1}
	Suppose we mistakenly redefine the connector mapping $\conn$ in Ex.~\ref{ex:motivating+conn}
	as follows (cf.\ Ex.~\ref{ex:motivating+flex}; i.e., the boxed label is wrong):
	\setlength{\fboxsep}{1.5pt}
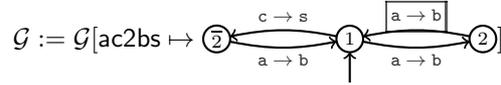
	\begin{equation*}
		\conn := \conn [\med{ac2bs} \mapsto \raisebox{.67ex}{\begin{tikzautom}
			\state{S1}{0,0}{\co2}
			\state{S2}{right of=S1}{1}
			\state{S3}{right of=S2}{2}
			\trans{S1}{S2}{from16,to9}{below}{\pid a \flow \pid b}
			\trans{S2}{S1}{from8,to1}{above}{\pid c \flow \pid s}
			\trans{S2}{S3}{from16,to9}{below}{\pid a \flow \pid b}
			\init{below of}{S2}
			\trans{S3}{S2}{from8,to1}{above}{\smash{\boxed{\strut\pid a \flow \pid b}}}
		\end{tikzautom}}]
	\end{equation*}
	Thus, connector $\med{ac2bs}$ initially allows a communication either between Alice and the bank, or between Carol and the shipper.
	In the latter case, $\med{ac2bs}$ subsequently allows a communication between Alice and the bank, as in Ex.~\ref{ex:motivating+flex}.
	But in the former case, $\med{ac2bs}$ subsequently allows a second communication between Alice and the bank (instead of between Carol and the shipper).

	The first derivation in Ex.~\ref{ex:motivating+flex} is still valid, but the second derivation is not: rule \rname{C}{SyncVal} is still applied to derive $\{ \com{\pid c\langle\mathit{book}\rangle}{\pid s} \}, \sigma', \emptyset \tr[\pid c \flow \pid s] \emptyset, \sigma'', \emptyset$ to fulfill the second premise of rule \rname{C}{Com}, but $\med{ac2bs}$ has no transition in state $2$ labelled with $\pid c \flow \pid s$.
	As there are no other derivations to fulfill the second premise of rule \rname{C}{Com}, the choreography gets stuck.
	\qed
\end{example}

\begin{example}[Book sale]\label{ex:motivating+resp2}
	Suppose we mistakenly redefine connector mapping $\conn$ in Ex.~\ref{ex:motivating+conn} as follows (i.e., the boxed process names are wrong\slash swapped):
	\setlength{\fboxsep}{1.5pt}
	\begin{equation*}
		\conn := \conn [\med{ac2bs} \mapsto \reo{Barrier}[\pid a / p_1, \smash{\boxed{\strut\pid s}} / p_2, \pid c / p_3, \smash{\boxed{\strut\pid b}} / p_4]]
	\end{equation*}
	Formally, automaton $\conn(\med{ac2bs})$ has the following transition: $1 \tr[\pid a \flow \pid s \wedge \pid c \flow \pid b]_{\med{ac2bs}} 1$.
	Thus, connector $\med{ac2bs}$ allows communications between Alice and the shipper (instead of the bank), and between Carol and the bank (instead of the shipper).

	The derivation in Ex.~\ref{ex:motivating+reduc+ac2bs} is no longer valid: rules \rname{C}{SyncVal} and \rname{C}{Join} are still applied to derive $\{ \com{\pid a\langle\mathit{money}\rangle}{\pid b}, \com{\pid c\langle\mathit{book}\rangle}{\pid s} \}, \sigma, \emptyset \tr[\pid a \flow \pid b \wedge \pid c \flow \pid s] \emptyset, \sigma'', \emptyset$ to fulfill the second premise of rule \rname{C}{Com}, but $\med{ac2bs}$ has no transition labelled with $\pid a \flow \pid b \wedge \pid c \flow \pid s$.
	As there are no other derivations to fulfill the second premise of rule \rname{C}{Com}, the choreography gets stuck.
	\qed
\end{example}

\begin{example}[Book sale]\label{ex:motivating+resp3}
	Suppose we mistakenly redefine the connector mapping $\conn$ in Ex.~\ref{ex:motivating+conn} as follows (i.e., the boxed process names are wrong\slash swapped):
	\setlength{\fboxsep}{1.5pt}
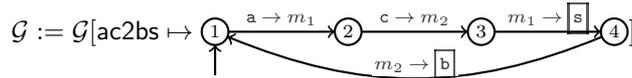
	\begin{equation*}
		\conn := \conn [\med{ac2bs} \mapsto \raisebox{.67ex}{\begin{tikzautom}
			\state{S1}{0,0}{1}
			\state{S2}{right of=S1}{2}
			\state{S3}{right of=S2}{3}
			\state{S4}{right of=S3}{4}
			\init{below of}{S1}
			\trans{S1}{S2}{}{above}{\pid a \flow m_1}
			\trans{S2}{S3}{}{above}{\pid c \flow m_2}
			\trans{S3}{S4}{}{above}{m_1 \flow \smash{\boxed{\strut\pid s}}}
			\trans{S4}{S1}{out=-157.5, in=-22.5}{above}{m_2 \flow \smash{\boxed{\strut\pid b}}}
		\end{tikzautom}}]
	\end{equation*}
	Thus, connector $\med{ac2bs}$ allows an asynchronous send by Alice, followed by an asynchronous send by Carol, followed by an asynchronous receive by the shipper, followed by an asynchronous receive by the bank.
	However, the shipper receives the value sent by Alice (instead of Carol), while the bank receives the value sent by Carol (instead of Alice)

	Let $\sigma$, $\sigma'$, $\sigma''$, and $\astate$ be defined as in Ex.~\ref{ex:motivating+reduc+ac2bs}.
	Furthermore, let $\mu = \{m_1 \mapsto \bot, m_2 \mapsto \bot \}$, let $\mu' = \{m_1 \mapsto \code{\$10}, m_2 \mapsto \bot \}$, and let $\mu'' = \{m_1 \mapsto \code{\$10}, m_2 \mapsto \code{foo.pdf} \}$.
	The following reductions can be derived using rule \rname{C}{Com}:
	\begin{alignat*}{4}
		&
	&	\quad\{ \com{\pid a\langle\mathit{money}\rangle}{\pid b}, \com{\pid c\langle\mathit{book}\rangle}{\pid s} \}\thru\med{ac2bs};\nil, &
	&	&\quad \sigma,
	&	&\quad \astate [\med{ac2bs} \mapsto \langle 1, \mu \rangle]
	\\	{\rightsquigarrow_\conn} &
	&	\quad \{ \pid b.\mathit{money}?\code{\$10}, \com{\pid c\langle\mathit{book}\rangle}{\pid s} \}\thru\med{ac2bs};\nil, &
	&	&\quad \sigma,
	&	&\quad \astate [\med{ac2bs} \mapsto \langle 2, \mu' \rangle]
	\\	{\rightsquigarrow_\conn} &
	&	\quad \{ \pid b.\mathit{money}?\code{\$10}, \pid s.\mathit{book}?\code{foo.pdf} \}\thru\med{ac2bs};\nil, &
	&	&\quad \sigma,
	&	&\quad \astate [\med{ac2bs} \mapsto \langle 3, \mu'' \rangle]
	\end{alignat*}
	At this point, the choreography gets stuck: there are no derivations to fulfill the second premise of \rname{C}{Com}.
	To see this, note that only rule \rname{C}{RecvVal} may be applicable (together with rule \rname{C}{Join}), but $\mu''(m_2) = \code{foo.pdf}$, whereas the choreography states $\pid b.\mathit{money}?\code{\$10}$.
	In other words, the choreography expects the bank to receive \code{\$10}, but connector $\med{ac2bs}$ allows the bank only to receive \code{foo.pdf}, out of memory cell $m_2$.
	\qed
\end{example}

None of these situations can arise
in existing choreography models, where all (a)synchronous channels are guaranteed to respect their choreographies, because the choreography syntax is carefully tuned to the \emph{fixed} communication semantics of these channels.
In CR, we have no fixed communication semantics: the fact that connectors in CR may not respect their choreography is, thus, a consequence of the added expressiveness and flexibility CR provides.

We proceed with a more formal account.

\begin{definition}
  \label{defn:respects}
  Connector mapping $\conn$ in automaton state function $\astate$ \emph{respects} choreography $C$ if: for every
  $\sigma$, $\tilde\eta$, $\gamma$, $\sigma'$ and $\astate'$, if
  $C,\sigma,\astate\redM^\ast\tilde\eta\thru\gamma;C',\sigma',\astate'$, then there exist $\sigma''$ and $\astate''$ such
  that $\tilde\eta\thru\gamma;C',\sigma',\astate'\redM^\ast C',\sigma'',\astate''$.
  Connector mapping $\conn$ respects choreography $C$ if $\conn$ respects $C$ in initial automaton state function $\astate_0$ (which assigns each automaton to its initial state and memory snapshot, as specified in $\conn$).
\end{definition}


\begin{definition}
  $C, \sigma, \astate_0$ is deadlock-free for every $\sigma$ iff $\conn$ respects $C$.
\end{definition}

\noindent
We can show respectfulness/deadlock-freedom to be undecidable using a classical
recursion-theoretic argument.

\begin{theorem}[Undecidability of Deadlock-Freedom]
\label{thm:respect-undec}
In general, it is undecidable whether a connector mapping $\conn$ respects a choreography $C$.
\end{theorem}
\begin{proof}
  Let $\eta$ be a communication action that does not respect $\astate$, and assume that connector $\gamma$ has
  synchronous links $\med{p2q}$ and $\med{q2p}$, from $\pid p$ to $\pid q$ and conversely
  (e.g., \reo{Sync} in Fig.~\ref{fig:aut}).
  Synchronous links are always enabled and do not change $\astate$.

  Let $f$ be a total function implemented at $\pid p$ and consider the choreography
  \begin{align*}
  C \equiv{} & \m{def}\ X=\valcom{\pid q}y{\pid p}x\thru\gamma;\\
  & \phantom{\m{def}\ X={}}                                      
  \m{if}\ (\pid p.f(x)=0)\ \m{then}\ \valcom{\pid p}{x+1}{\pid q}y\thru\gamma;X\\
  & \phantom{\m{def}\ X=\m{if}\ (\pid p.f(x)=0)\ }         
  \m{else}\ \eta\thru\gamma\\
  & \m{in}\ \valcom{\pid p}{0}{\pid q}y\thru\gamma;\ X
  \end{align*}

  In this choreography, $\pid q$ sequentially sends the natural numbers to $\pid p$, which applies $f$ to its input and
  proceeds if the result is $0$.
  If $\pid q$ sends a value where $f$ is not $0$, the choreography attempts to perform $\eta$ and
  deadlocks.
  Then $C$ respects $\astate$ iff $f$ is constantly equal to $0$, which by Rice's theorem is not decidable.
\end{proof}

\begin{remark}
Undecidability of deadlock-freedom arises because of the
new ways in which a choreography and connectors can affect each other, which did not
exist in previous work.
Specifically, deadlock occurs if a connector's current state has no transitions
for the interactions in the choreography's current $\tilde\eta$. In previous
models, this can never happen, since the choreography syntax matches the
hard-wired communication semantics \emph{by definition}. Violation of
respectfulness is, thus, a unique byproduct of allowing custom communication
semantics, through connectors. Concretely, the proof of
Thm.~\ref{thm:respect-undec} relies on the existence of a communication action
that does not respect $\conn$, which does not exist in previous models.
\qed
\end{remark}

We can approximate respectfulness by a decidable relation, called \emph{compatibility}, essentially by abstracting
away from data.
The key point is that a conditional satisfies compatibility only if both its branches satisfy
compatibility.

\begin{definition}
  \label{defn:compat}

\begin{figure}[t]
  \small
  \begin{eqnarray*}
    &\infer[\rname{CC}{Com}]
    {
      \Gamma\vdash^\conn_\astate C
    }
    {
    \astate(\gamma) = \langle s,\mu\rangle
    &
      \left\{
      \Gamma\vdash^\conn_{\astate[\gamma\to\langle s',\mu'\rangle]}\tilde\eta'\thru\gamma;C'
      \left|
      \begin{gathered}
      C\equiv\tilde\eta\thru\gamma;C'
      \\\wwedge
      \tilde\eta,\mu\tr[\phi]\tilde\eta',\mu'
      \\\wwedge
      s\tr[\phi]_\gamma s'
      \end{gathered}
      \right.\right\}
      & (\dagger)
    }
    \\[1ex]
    &\infer[\rname{CC}{Done}]
    {
      \Gamma\vdash^\conn_\astate \emptyset\thru\gamma;C
    }
    {
      \Gamma\vdash^\conn_\astate C
    }
    \qquad
	\infer[\rname{CC}{Nil}]
    {
      \Gamma\vdash^\conn_\astate\nil
    }
    {}
    \qquad
    \infer[\rname{CC}{Cond}]
    {
      \Gamma\vdash^\conn_\astate\gencond
    }
    {
      \Gamma\vdash^\conn_\astate C_1
      &
      \Gamma\vdash^\conn_\astate C_2
    }
    \\[1ex]
    &\infer[\rname{CC}{Def}]
    {
      \Gamma\vdash^\conn_\astate\genrec
    }
    {
      \Gamma,(X:\astate_X)\vdash^\conn_\astate C_1
      &
      \Gamma,(X:\astate_X)\vdash^\conn_\astate C_2
    }
    \qquad
    \infer[\rname{CC}{Call}]
    {
      \Gamma\vdash^\conn_\astate X
    }
    {
      (X:\mathcal A) \in \Gamma
    }
  \end{eqnarray*}
  \caption{Cho-Reo-graphy, compatibility relation.
    The side condition $(\dagger)$ reads: the set of judgments on the left is
    nonempty. We abuse notation in rule \rname{CC}{Com} to indicate that all
    judgments in this set must be true.}
  \label{fig:compatibility}
\end{figure}

  Let $C$ be a choreography, $\conn$ be a connector mapping, and $\astate$ be an automaton state function.
  We say that $C$ and $\conn$ are \emph{compatible} by automaton state function $\astate$ if $\vdash^\conn_\astate C$, where the
  relation $\vdash$ is defined by the rules in Fig.~\ref{fig:compatibility}.
  We say $C$ and $\conn$ are compatible, written $\vdash^\conn C$, if $\vdash^\conn_{\astate_0}C$ with $\astate_0$ as in Definition~\ref{defn:respects}.

  Relation $\vdash$ uses a context $\Gamma$, defined inductively as
 $\Gamma \mathrel{::=} (X:\astate), \Gamma \mid \cdot $,
  and an abstraction of the labelled reductions for communications from Fig.~\ref{fig:compat},
$\tilde\eta,\mu\tr[\phi]\tilde\eta',\mu'$. The latter models a symbolic execution of communications; it is defined as
in Fig.~\ref{fig:compat}, with two differences: (i)~$\sigma$ is removed from the domain of
the reduction and (ii)~in rule \rname{C}{SendVal}, $v$ is a fresh token.
\end{definition}

We comment on this relation.
Intuitively, $X:\astate \in \Gamma$ indicates that procedure $X$ can be called only whenever the automata have current states
$\astate$; this is encoded in rules $\rname{CC}{Def}$ and $\rname{CC}{Call}$ (in the former rule, a unique automaton state
function $\astate_X$ is stipulated; in the latter rule, it is checked against the current automaton state function $\astate$).
Together with the fact that we allow actions to be swapped in rule \rname{CC}{Com}, but not recursive calls to be unfolded,
this means that the recursive structures of the choreography and the automata in the connector mapping must be similar (i.e., the loops in the automata must match the recursions in the choreography).
Furthermore, in order for these rules to ensure respectfulness, the transition relation in the
automaton also needs to be confluent (cf.\ Thm.~\ref{thm:compat-respect}).

\begin{remark}
  Compatibility can become more robust/modular by disregarding connectors
  not occurring in procedure bodies in $\rname{CC}{Def}$.
  We chose the current formulation for simplicity.
  \qed
\end{remark}

We now revisit our previous examples to demonstrate cases where the compatibility relation
constitutes a precise approximation of respectfulness.

\begin{example}[Book sale]
  \label{ex:respect1}
	We illustrate how compatibility works in the context of our running example by revisiting choreography $\{ \com{\pid a\langle\mathit{money}\rangle}{\pid b},
	\com{\pid c\langle\mathit{book}\rangle}{\pid s} \}\allowbreak\thru\med{ac2bs};\nil$ with five
different connector mappings from previous examples.

	\begin{itemize}
		\item[\yes]
		Let $\conn$ and $\astate$ be defined as in Ex.~\ref{ex:motivating+reduc+ac2bs}.
		Using Fig.~\ref{fig:compatibility}, we derive:
		\begin{equation*}
			\dfrac{%
				\raisebox{-10pt}{$\astate(\med{ac2bs}) = \langle 1, \emptyset \rangle$}
				\qquad
				\dfrac{%
					\dfrac{%
					}{%
						{\cdot} \vdash_{\astate[\med{ac2bs}\mapsto\langle1,\emptyset\rangle]}^\conn \nil
					}
					\raisebox{.5pt}{\rlap{\scriptsize\rname{CC}{Nil}}}
				}{%
					{\cdot} \vdash_{\astate[\med{ac2bs}\mapsto\langle1,\emptyset\rangle]}^\conn \emptyset\thru\med{ac2bs};\nil
				}
				\raisebox{.5pt}{\rlap{\scriptsize\rname{CC}{Done}}}
			}{%
				{\cdot} \vdash_\astate^\conn \{ \com{\pid a\langle\mathit{money}\rangle}{\pid b}, \com{\pid c\langle\mathit{book}\rangle}{\pid s} \}\thru\med{ac2bs};\nil
			}
			\raisebox{.5pt}{\rlap{\scriptsize\rname{CC}{Com}}}
		\end{equation*}
		Thus, the choreography and the connector mapping are compatible.
                Corollary~\ref{cor:compat-sound} below
                implies that the connector mapping respects the choreography.

		\item[\yes]
		Let $\conn$ and $\astate$ be defined as in Ex.~\ref{ex:motivating+flex}.
		Furthermore, let $\astate_2' = \astate[\med{ac2bs}\mapsto\langle2,\emptyset\rangle]$, let $\astate_{\co2}' = \astate[\med{ac2bs}\mapsto\langle\co2,\emptyset\rangle]$, and let $\astate'' = \astate_2'[\med{ac2bs}\mapsto\langle1,\emptyset\rangle] = \astate_{\co2}'[\med{ac2bs}\mapsto\langle1,\emptyset\rangle] = \astate$.
		Using Fig.~\ref{fig:compatibility}, we derive:
		\begin{equation*}
			\dfrac{%
				\raisebox{-14pt}{$\begin{array}{@{}c@{}}
					\astate(\med{ac2bs})
				\\	{} = \langle 1, \emptyset \rangle
				\end{array}$}
				\qquad
				\dfrac{%
					\raisebox{-14pt}{$\begin{array}{@{}c@{}}
						\astate_{\co2}'(\med{ac2bs})
					\\	{} = \langle\co2,\emptyset\rangle
					\end{array}$}
					\qquad
					\dfrac{%
						\dfrac{%
						}{%
							{\cdot} \vdash_{\astate''}^\conn \nil
						}
					}{%
						\begin{array}{@{}c@{}}
							{\cdot} \vdash_{\astate''}^\conn \emptyset\thru
						\\	\med{ac2bs};\nil
						\end{array}
					}
				}{%
					{\cdot} \vdash_{\astate_{\co2}'}^\conn \begin{array}[t]{@{}c@{}}
						\{\com{\pid a\langle\mathit{money}\rangle}{\pid b}\}
					\\	\thru\med{ac2bs};\nil
					\end{array}
				}
				\qquad
				\dfrac{%
					\raisebox{-14pt}{$\begin{array}{@{}c@{}}
						\astate_{2}'(\med{ac2bs})
					\\	{} = \langle2,\emptyset\rangle
					\end{array}$}
					\qquad
					\dfrac{%
						\dfrac{%
						}{%
							{\cdot} \vdash_{\astate''}^\conn \nil
						}
						\raisebox{.5pt}{\rlap{\scriptsize\rname{CC}{Nil}}}
					}{%
						\begin{array}{@{}c@{}}
							{\cdot} \vdash_{\astate''}^\conn \emptyset\thru
						\\	\med{ac2bs};\nil
						\end{array}
					}
					\raisebox{.5pt}{\rlap{\scriptsize\rname{CC}{Done}}}
				}{%
					{\cdot} \vdash_{\astate_2'}^\conn \begin{array}[t]{@{}c@{}}
						\{\com{\pid c\langle\mathit{book}\rangle}{\pid s}\}
					\\	\thru\med{ac2bs};\nil
					\end{array}
				}
				\raisebox{.5pt}{\rlap{\scriptsize\rname{CC}{Com}}}
			}{%
				{\cdot} \vdash_\astate^\conn \{ \com{\pid a\langle\mathit{money}\rangle}{\pid b}, \com{\pid c\langle\mathit{book}\rangle}{\pid s} \}\thru\med{ac2bs};\nil
			}
			\raisebox{.5pt}{{\scriptsize\rname{CC}{Com}}}
		\end{equation*}
		The bottom application of rule \rname{CC}{Com} requires two subderivations: one to cover the case where connector $\med{ac2bs}$ makes a transition to state $\co2$ (left subderivation), and another to cover the case where $\med{ac2bs}$ makes a transition to state $2$ (right subderivation).
		In both cases, we have compatibility.

		Thus, the choreography and the connector mapping are compatible.
                Corollary~\ref{cor:compat-sound} below
                implies that the connector mapping respects the choreography.

		\item[\no]
		Let $\conn$ and $\astate$ be defined as in Ex.~\ref{ex:motivating+resp1}.
		Furthermore, let $\astate_2' = \astate[\med{ac2bs}\mapsto\langle2,\emptyset\rangle]$, and let $\astate_{\co2}' = \astate[\med{ac2bs}\mapsto\langle\co2,\emptyset\rangle]$.
		Using Fig.~\ref{fig:compatibility}, we attempt:
		\begin{equation*}
			\dfrac{%
				\raisebox{-14pt}{$\begin{array}{@{}c@{}}
					\astate(\med{ac2bs})
				\\	{} = \langle 1, \emptyset \rangle
				\end{array}$}
				\qquad
				\dfrac{%
					\raisebox{-14pt}{$\begin{array}{@{}c@{}}
						\astate_{\co2}'(\med{ac2bs})
					\\	{} = \langle\co2,\emptyset\rangle
					\end{array}$}
					\qquad
					\dfrac{%
						\dfrac{%
						}{%
							{\cdot} \vdash_{\astate''}^\conn \nil
						}
					}{%
						\begin{array}{@{}c@{}}
							{\cdot} \vdash_{\astate''}^\conn \emptyset\thru
						\\	\med{ac2bs};\nil
						\end{array}
					}
				}{%
					{\cdot} \vdash_{\astate_{\co2}'}^\conn \begin{array}[t]{@{}c@{}}
						\{\com{\pid a\langle\mathit{money}\rangle}{\pid b}\}
					\\	\thru\med{ac2bs};\nil
					\end{array}
				}
				\qquad
				\dfrac{%
					\raisebox{0pt}{$\begin{array}{@{}c@{}}
						\astate_{2}'(\med{ac2bs})
					\\	{} = \langle2,\emptyset\rangle
					\end{array}$}
					\qquad
					\setlength{\fboxsep}{-1pt}
					\boxed{\phantom{
						\begin{array}{@{}c@{}}
							\strut\\\strut
						\end{array}\qquad\qquad\qquad
					}}
				}{%
					{\cdot} \vdash_{\astate_2'}^\conn \begin{array}[t]{@{}c@{}}
						\{\com{\pid c\langle\mathit{book}\rangle}{\pid s}\}
					\\	\thru\med{ac2bs};\nil
					\end{array}
				}
				\raisebox{.5pt}{\rlap{\scriptsize\rname{CC}{Com}}}
			}{%
				{\cdot} \vdash_\astate^\conn \{ \com{\pid a\langle\mathit{money}\rangle}{\pid b}, \com{\pid c\langle\mathit{book}\rangle}{\pid s} \}\thru\med{ac2bs};\nil
			}
			\raisebox{.5pt}{{\scriptsize\rname{CC}{Com}}}
		\end{equation*}
		This attempted derivation is the same as in the second \yes-item, except the subderivation inside the box has become invalid: under our current $\conn$, connector $\med{ac2bs}$ in state $2$ has no transition labelled with $\pid c \flow \pid s$.

		Thus, the choreography and the connector mapping are incompatible.
		In fact, in this case, the choreography may deadlock under $\conn$.

		\item[\no]
		Let $\conn$ and $\astate$ be defined as in Ex.~\ref{ex:motivating+resp2}.
		Using Fig.~\ref{fig:compatibility}, we attempt:
		\begin{equation*}
			\dfrac{%
				\raisebox{-10pt}{$\astate(\med{ac2bs}) = \langle 1, \emptyset \rangle$}
				\qquad
				\setlength{\fboxsep}{-1pt}
				\boxed{\phantom{
					\begin{array}{@{}c@{}}
						\strut\\\strut
					\end{array}\qquad\qquad\qquad
				}}
			}{%
				{\cdot} \vdash_\astate^\conn \{ \com{\pid a\langle\mathit{money}\rangle}{\pid b}, \com{\pid c\langle\mathit{book}\rangle}{\pid s} \}\thru\med{ac2bs};\nil
			}
			\raisebox{.5pt}{\rlap{\scriptsize\rname{CC}{Com}}}
		\end{equation*}
		This attempted derivation is the same as in the first \yes-item, except the subderivation inside the box has become invalid: under our current $\conn$, connector $\med{ac2bs}$ in state $1$ has no transition labelled with $\pid a \flow \pid b \wedge \pid c \flow \pid s$.

		Thus, the choreography and the connector mapping are incompatible.
		In fact, in this case, the choreography deadlocks under $\conn$.

		\item[\no]
		Let $\conn$ and $\astate$ be defined as in Ex.~\ref{ex:motivating+resp3}.
		Furthermore, let $\diamondsuit$ and $\heartsuit$ denote two fresh token values, let $\mu = \{m_1\mapsto\bot,m_2\mapsto\bot\}$, let $\mu' = \{m_1\mapsto\diamondsuit,m_2\mapsto\bot\}$, let $\mu'' = \{m_1\mapsto\diamondsuit,m_2\mapsto\heartsuit\}$, let $\astate' = \astate[\med{ac2bs} \mapsto \langle2,\mu'\rangle]$, and let $\astate'' = \astate[\med{ac2bs} \mapsto \langle2,\mu''\rangle]$.
		Using Fig.~\ref{fig:compatibility}, we attempt:
		\begin{equation*}
			\dfrac{%
				\raisebox{-14pt}{$\begin{array}{@{}c@{}}
					\astate(\med{ac2bs})
				\\	{} = \langle 1, \mu \rangle
				\end{array}$}
				\qquad
				\dfrac{%
					\raisebox{-14pt}{$\begin{array}{@{}c@{}}
						\astate'(\med{ac2bs})
					\\	{} = \langle 2, \mu' \rangle
					\end{array}$}
					\qquad
					\dfrac{%
						\raisebox{0pt}{$\begin{array}{@{}c@{}}
							\astate''(\med{ac2bs})
						\\	{} = \langle 3, \mu'' \rangle
						\end{array}$}
						\qquad
						\setlength{\fboxsep}{-1pt}
						\boxed{\phantom{
							\begin{array}{@{}c@{}}
								\strut\\\strut
							\end{array}\qquad\qquad\qquad
						}}
					}{%
						{\cdot} \vdash_{\astate''}^\conn \begin{array}[t]{@{}c@{}}
							\{ \pid b.\mathit{money}?\diamondsuit, \pid s.\mathit{book}?\heartsuit \}
						\\	\thru\med{ac2bs};\nil
						\end{array}
					}
					\raisebox{.5pt}{\rlap{\scriptsize\rname{CC}{Com}}}
				}{%
					{\cdot} \vdash_{\astate'}^\conn \begin{array}[t]{@{}c@{}}
						\{ \pid b.\mathit{money}?\diamondsuit, \com{\pid c\langle\mathit{book}\rangle}{\pid s} \}
					\\	\thru\med{ac2bs};\nil
					\end{array}
				}
				\raisebox{.5pt}{\rlap{\scriptsize\rname{CC}{Com}}}
			}{%
				{\cdot} \vdash_{\astate}^\conn \{ \com{\pid a\langle\mathit{money}\rangle}{\pid b}, \com{\pid c\langle\mathit{book}\rangle}{\pid s} \}\thru\med{ac2bs};\nil
			}
			\raisebox{.5pt}{{\scriptsize\rname{CC}{Com}}}
		\end{equation*}
		This attempted derivation fails, because the intended subderivation inside box (the receive of $\diamondsuit$, followed by the receive of $\heartsuit$) is invalid: under our current $\conn$, connector $\med{ac2bs}$ in state 3 has no transition labelled with $m_1 \flow \pid b$.

		Thus, the choreography and the connector mapping are incompatible.
		In fact, in this case, the choreography deadlocks under $\conn$.
		\qed
	\end{itemize}
\end{example}

However, the restriction that a choreography and the automata in its connector mapping must have similar recursive structures (for them to be judged compatible) implies there exist connector mappings that respect their choreographies, but that
cannot be shown to do so by means of the compatibility relation -- which is unavoidable in view of our undecidability
result.
We illustrate this by some examples.

\begin{example}
  \label{ex:respect2}
  Let $C$ be the simple choreography: $$\rec{X}{\com{\pid p}{\pid q}\thru\gamma;\com{\pid p}{\pid q}\thru\gamma;\com{\pid r}{\pid s}\thru\gamma;X}{X}$$ and $\conn(\gamma)$ a connector that allows communications from $\pid p$ to $\pid q$ to occur simultaneously with communications from $\pid r$ to $\pid s$ (e.g., $\reo{Barrier}$ in Fig.~\ref{fig:aut}).
  Then $C$ is deadlock-free, since structural precongruence allows the second communication from
  $\pid p$ to $\pid q$ to be ``delayed'' and the communication from $\pid r$ to $\pid s$ to be
  ``pushed forward'':
  \begin{align*}
    &
    \com{\pid p}{\pid q}\thru\gamma;\com{\pid p}{\pid q}\thru\gamma;\com{\pid r}{\pid s}\thru\gamma \\
    \equiv{} &
    \com{\pid p}{\pid q}\thru\gamma;\{\com{\pid p}{\pid q},\com{\pid r}{\pid s}\}\thru\gamma
    & \mbox{by} & \rname{C}{Eta-Split}\\
    \equiv{} &
    \com{\pid p}{\pid q}\thru\gamma;\com{\pid r}{\pid s}\thru\gamma;\com{\pid p}{\pid q}\thru\gamma
    & \mbox{by} & \rname{C}{Eta-Split}
  \end{align*}
  However, $\not\vdash^\conn C$, since the second communication from $\pid p$ to $\pid q$ in the body of $X$ cannot be consumed without unfolding the definition of $X$.

  In this example, the recursive structure of $X$ (two communications from $\pid p$ to $\pid q$ and one from $\pid r$ to $\pid s$) differs from the recursive structure of $\conn(\gamma)$ (one communication between each pair of processes).
  \qed
\end{example}

\begin{example}
  \label{ex:respect3}
  Consider now the choreography $C$ defined as $$\rec{X}{\com{\pid p}{\pid q}\thru\gamma;\com{\pid r}{\pid s}\thru\gamma;X}{X}$$ where $\conn(\gamma)$ only allows communications from $\pid p$ to $\pid q$.
  Again $C$ is deadlock-free, since structural precongruence allows the communications from $\pid r$ to $\pid s$ to be indefinitely postponed.
  However, $\not\vdash^\conn C$.
  In general, compatibility ensures that the choreography is not only deadlock-free,
  but also that there is a correspondence between the recursive structure of the choreography
    and the recursive structure of the connectors: the connector must allow all interactions in the body
    of a definition to be executed before calling other procedures.
  \qed
\end{example}



\begin{theorem}[Compatibility Preservation]
  \label{thm:compat-respect}
  Let $C$ and $C'$ be choreographies, $\conn$ be a connector mapping,
  $\sigma$ and $\sigma'$ be choreography states, and $\astate$ and $\astate'$ be automaton state functions.
  If the transition relation in each $\conn(\gamma)$ is confluent,
  $\vdash^\conn_\astate C$ and $C,\sigma,\astate\redM C',\sigma',\astate'$, then $\vdash^\conn_{\astate'} C'$.
\end{theorem}
\begin{proof}
  Straightforward by case analysis on the reduction from $C,\sigma,\astate$ to $C',\sigma',\astate'$, using the fact that the automata are confluent (to make sure unfolding cannot add unwanted reductions), and therefore compatibility is preserved by structural precongruence.
\end{proof}

The hypothesis that the transition relations of automata are confluent is required to make
sure that unfolding cannot add unwanted reductions.

\begin{corollary}[Soundness of Compatibility]
\label{cor:compat-sound}
  Under the assumptions of Theorem \ref{thm:compat-respect}, if
  $\vdash^\conn_\astate C$, then $\conn$ in $\astate$ respects $C$.
\end{corollary}
\begin{proof}
  If $C,\sigma\astate\redM^\ast\eta;C',\sigma',\astate'$, then, by induction on the length of this sequence of reductions, we use
  Theorem~\ref{thm:compat-respect} to show that $\eta;C',\sigma',\astate'\redM C',\sigma'',\astate''$ for some $\astate''$ and
  $\sigma''$.
\end{proof}

Furthermore, compatibility is decidable.

\begin{theorem}[Decidability of Compatibility]
\label{thm:compat-decide}
  There is an algorithm that, given $C$, $\conn$ and $\astate$, returns \textsc{yes} if $\vdash^\conn_\astate C$ and
  \textsc{no} if $\not\vdash^\conn_\astate C$ in time
  $O(P \times \max p_A \times k \times (\sum d_A) ^ {2k})$, where $p_A$ is the maximum number of ports in any
  automaton, $d_A$ is the maximum number of transitions from a state in each automaton, $k$ is the maximum
  number of communication actions in any procedure definition (or main choreography) and $P$ is the total
  number of procedure definitions (including the main choreography).
\end{theorem}

A simple finiteness argument suffices for establishing decidability of
compatibility, since the number of automaton states is finite, the number of
applicable rules at each step is finite, and all rules have a finite number of
premises, and the size of the choreographies in the premises is always smaller
than the size of the choreographies in the conclusions.
Therefore, by non-deterministically guessing the types of all procedures, we can
decide whether $\vdash^\conn_\astate C$ or not. However, we provide a more
intelligent proof that constructs the types for recursive definitions.

\begin{proof}[Proof of Theorem~\ref{thm:compat-decide}]
  We assume that every procedure defined in $C$ is called at least once outside of its body.

  The idea behind our algorithm is to construct a derivation for $\vdash^\conn_\astate C$ by applying
  the rules in Fig.~\ref{fig:compatibility} bottom-up.
  When we meet a term of the form $\genrec$, we focus on $C_1$ first, and leave $\astate_X$ (see rule $\rname{CC}{Def}$)
  unspecified. We instantiate $\astate_X$ later, when we meet $X$ for the first time inside of $C_1$.
  More precisely:
  \begin{enumerate}
  \item Initialize a list $\mathcal L=[\cdot\vdash^\conn_\astate C]$.
  \item While $\mathcal L$ is not empty:
    \begin{enumerate}
    \item Remove the first pending judgement $\Gamma\vdash^\conn_\astate C$ from $\mathcal L$.
    \item If $C$ is $\nil$, proceed to the next iteration.
    \item If $C$ is of the form $\gencond$, then add $\Gamma\vdash^\conn_\astate C_1$ and
      $\Gamma\vdash^\conn_\astate C_2$ at the beginning of $\mathcal L$.
    \item If $C$ is of the form $\genrec$, then add $\Gamma,(X:\astate_X)\vdash^\conn_\astate C_1$ and
      $\Gamma,(X:\astate_X)\vdash^\conn_\astate C_2$, \emph{in this order}, at the \emph{beginning} of
      $\mathcal L$.
      Here, $\astate_X$ is a unique variable representing an unknown state function.
    \item If $C$ is of the form $X$, there are two cases.
      If $\Gamma$ contains $(X:\astate_X)$ with $\astate_X$ instantiated, check whether $\astate_X=\astate$;
      if so, proceed to the next iteration, otherwise return \textsc{no}.
      If $\Gamma$ contains $(X:\astate_X)$ with $\astate_X$ uninstantiated, replace all occurrences of
      $\astate_X$ in $\mathcal L$ by $\astate$ and proceed to the next iteration.\footnote{Note that $\Gamma$
        must contain $(X:\astate_X)$, otherwise the initial choreography is not well-formed.}
    \item Otherwise, $C$ is of the form $\tilde\eta\thru\gamma;C'$.
      Consider all possible ways of rewriting $C$ as $\tilde\eta'\thru\gamma;C'$ by swapping
      independent actions, without unfolding recursive definitions.
      Let $\astate(\gamma)=\langle s,\mu\rangle$.
      For each such $\tilde\eta'$, check whether $\tilde\eta',\mu\tr[\phi]\tilde\eta'',\mu'$ for some
      $\phi$, and in the affirmative case compute $s'$ such that $s\tr[\phi]_\gamma s'$ and add
      $\Gamma\vdash^\conn_{\astate[\gamma\to\langle s',\mu'\rangle]}\tilde\eta''\thru\gamma;C'$ at the beginning
      of $\mathcal L$.
      If no such transitions exist, return \textsc{no}.
    \end{enumerate}
  \item Return \textsc{yes}.
  \end{enumerate}

  Termination of this algorithm is straight-forward: the sum of the sizes of all the choreographies in $\mathcal L$
  stricly decreases at each iteration, and each step terminates in finite time.
  (The size of a choreography is the number of nodes in its abstract syntax tree, except that $\gencom$ and $\gensel$
  count as $2$, while $\genrecv$ and $\genchoice$ count as $1$.)
  Soundness is immediate by observing that the judgements stored in $\mathcal L$ are exactly those that are necessary to
  construct a derivation of $\vdash^\conn_\astate C$, since at each stage there is only one rule that can be applied to
build
  such a derivation, and this rule is determined by the structure of $C$.
  If the algorithm returns \textsc{yes}, then a valid derivation for $\vdash^\conn_\astate C$ can be built.
  If the algorithm returns \textsc{no} because of a mismatch between the state of the automata and a communication
action
  (Step 2.f), then clearly $\not\vdash^\conn_\astate C$.
  If the algorithm returns \textsc{no} because of an incompatibility between the state assigned to a procedure name $X$
  in $\Gamma$ and the state in the current judgement (Step 2.e), then this failure means that we constructed two
  judgements involving $X$ with different automaton state functions, which also implies that $\not\vdash^\conn_\astate
C$.

  To obtain the complexity bound, perform step 2f as follows: consider all possible transitions ($\sum d_A$)
  and check which ones are enabled (naively: go through the current choreography and check each transition,
  which yields $\max p_A \times k$).
  Each transition consumes at least half an interaction, so this can be repeated $\leq 2k$ times for each
  procedure definition.
\end{proof}

Since automata are deterministic, $d_A < 2^{p_A}$ (each subset of ports determines at most one transition),
and $p_A$ is at most the number of processes in the choreography. Thus, this upper bound can be stated
independently of the automata considered.

Although this complexity is high, we believe that compatibility checking is feasible in practice. The bound is
an over-approximation, since choreographies typically contain many causal dependencies among communications
that reduce non-determinism (as in our examples). Previous works on choreographies proposed algorithms with
even worse worst-case complexity, but feasible in practice for the same reason~\cite{LTY15}.

\begin{theorem}[Progress]
  \label{thm:progress}
  Let $C$ be a choreography, $\conn$ be a connector mapping, $\sigma$ be a choreography state and $\astate$ be an automaton state function
  such that $\vdash^\conn_\astate C$.
  Then, either $C\precongr\nil$ ($C$ has terminated) or there exist $C'$, $\sigma'$ and $\astate'$ such that
  $C,\sigma,\astate\redM C',\sigma',\astate'$.
\end{theorem}
\begin{proof}
  If $C\not\precongr\nil$, then $C$ is of the form $\tilde\eta\thru\gamma;C'$ or $\gencond$, eventually inside some
  recursive definitions.
  In the latter case, $C$ can always reduce; in the former case, compatibility guarantees that $C$ can reduce.
\end{proof}

\begin{theorem}[Deadlock-Freedom by Design]
  \label{thm:df}
  Let $C$ be a choreography, $\sigma$ be a choreography state function, and $\astate$ be an automaton state function.
  If $\vdash^\conn_\astate C$ and $C,\sigma,\astate\redM^\ast C',\sigma',\astate'$, then either $C'\precongr\nil$ or
  there exist $C''$, $\sigma''$ and $\astate''$ such that $C',\sigma',\astate'\redM C'',\sigma'',\astate''$.
\end{theorem}
\begin{proof}[Proof (sketch)]
  From Theorem~\ref{thm:progress}, if $C\not\precongr\nil$, then by Theorem~\ref{thm:compat-respect} we also
  have that $\vdash^\conn_{\astate'}C'$ whenever $C,\sigma,\astate\redM C',\sigma',\astate'$.
  The thesis then follows by induction.
\end{proof}

\section{Connected Processes}
\label{sec:processes}

CR shows how choreographies can be combined with connectors, but it does
not indicate how we can obtain executable implementations. The missing link is determining how a choreography can be compiled to terms representing executable processes that communicate through connectors.
We address this aspect by presenting a process calculus based on standard
I/O actions and a translation (compilation procedure) from CR to this calculus.

\subsection{Syntax and semantics}

We define Connected Processes (CP), the process calculus to represent concrete implementations of choreographies.
The syntax of CP is given in Fig.~\ref{fig:cp_syntax}.
A network $N$ is a parallel composition of processes. A process is written $\proc{\pid p}{\rho}{B}$, where $\pid p$ is
its identifier, $\rho$ its state (mapping variable names to values), and $B$ its behaviour.
Behaviours correspond to local views of choreography interactions. Procedure definitions and calls, conditionals, and
termination ($\nil$) follow the same ideas as in CR. Communication actions implement the local behaviour of
each process in a choreography interaction: sending a value through an output port ($\psend{\port o}{e}$); receiving
a value through an input port ($\precv{\port i}{x}$); selecting a label through an output port ($\psel{\port o}{\ell}$);
and offering a choice on some labels through an input port ($\pbranch{\port i}{\{ \ell_i : B_i\}_{i\in I}}$).

\begin{figure}[t]
\begin{align*}
B ::={} & \psend{\port o}{e};B \mid \precv{\port i}{x};B
\mid \psel{\port o}{\ell};B \mid
\pbranch{\port i}{\{ \ell_i : B_i\}_{i\in I}}
& N,M ::={} & \proc{\pid p}{\rho}{B} \mid (N \ppar M) \mid \nil
\\
{}\mid{} & \cond{e}{B_1}{B_2} \mid \rec{X}{B_2}{B_1} \mid X \mid \nil
\end{align*}
\caption{Connected Processes, Syntax.}
\label{fig:cp_syntax}
\end{figure}

The key difference from choreographies is that communications now refer to actual ports, instead of to connectors (we have no ``$\thru \gamma$'' for communications in the process calculus). This reflects the principle that processes should not know how they are connected~\cite{Arb11,Jon16,JA16}.

The semantics of CP is parameterised on connectors represented as a set of automata $\cconn$ that do not share
any ports. Differently from the automata used in CR, the ones in $\cconn$ use the names of the actual ports to which they are connected (and which are also used by the processes).
Reductions in CP have the form $N, A \toC N', A'$, where $A$ maps each automaton in $\cconn$ to a pair
$\langle s,\mu \rangle$ of its state $s$ and memory snapshot $\mu$.
The key reduction rule of CP is the one for communications.
\[
\infer[\rname{CP}{Com}]
      {
        N,A \toC N',A[a \mapsto \langle s',\mu'\rangle]
      }
      {
	a \in \cconn
	&
        A(a) = \langle s,\mu\rangle
        &
        N,\mu \tr[\phi] N',\mu'
        &
	s \tr[\phi]_a s'
      }
\]
This rule is reminiscent of rule $\rname{C}{Com}$ for choreographies.
In particular, it uses a similar auxiliary reduction relation on pairs of networks and memory snapshots (stated in premise
$N,\mu\tr[\phi] N',\mu'$), whose rules are given in Fig.~\ref{fig:proc-compat}.

\begin{figure}
  \small
  \begin{eqnarray*}
    &\infer[\rname{CP}{SyncVal}]
    {\proc{\pid p}{\rho_{\pid p}}{\psend{\port o}{e};B_{\pid p}}\ppar\proc{\pid q}{\rho_{\pid q}}{\precv{\port i}x;B_{\pid q}},\mu
      \tr[\pidp po\flow\pidp qi]
      \proc{\pid p}{\rho_{\pid p}}B_{\pid p}\ppar\proc{\pid q}{\rho_{\pid q}[x\mapsto v]}B_{\pid q},\mu\rule[1em]{0em}{0em}}
    {e\eval{\rho_{\pid p}}{}v}
  \\[1ex]
  &\infer[\rname{CP}{SendVal}]
    {\proc{\pid p}{\rho}{\psend{\port o}{e};B},\mu
      \tr[\pidp po\flow m]
      \proc{\pid p}{\rho}B,\mu[m\mapsto v]\rule[1em]{0em}{0em}}
    {e\eval{\rho}{}v}
  \\[1ex]
  &\infer[\rname{CP}{RecvVal}]
    {\proc{\pid q}{\rho}{\precv{\port i}x;B},\mu
      \tr[m\flow\pidp qi]
      \proc{\pid q}{\rho[x\mapsto v]}B,\mu\rule[1em]{0em}{0em}}
    {\mu(m)=v}
  \\[1ex]
  &\infer[\rname{CP}{SyncSel}]
    {\proc{\pid p}{\rho_{\pid p}}{\psel{\port o}{\ell_j};B}\ppar\proc{\pid q}{\rho_{\pid q}}{\pbranch{\port
i}{\{\ell_i:B_i\}_{i\in I}}},\mu
      \tr[\pidp po\flow\pidp qi]
      \proc{\pid p}{\rho_{\pid p}}B\ppar\proc{\pid q}{\rho_{\pid q}}B_j,\mu\rule[1em]{0em}{0em}}
    {j \in I}
  \\[1ex]
  &
  \infer[\rname{CP}{SendSel}]
    {\proc{\pid p}{\rho}{\psel{\port o}{\ell};B},\mu
      \tr[\pidp po\flow m]
      \proc{\pid p}{\rho}B,\mu\rule[1em]{0em}{0em}}
    {}
  \\[1ex]
  &\infer[\rname{CP}{RecvSel}]
    {\proc{\pid q}{\rho}{\pbranch{\port i}{\{\ell_i:B_i\}_{i\in I}}},\mu
      \tr[m\flow\pidp qi]
      \proc{\pid q}{\rho}B_j,\mu\rule[1em]{0em}{0em}}
    {\mu(m)=\ell_j & j \in I}
  \\[1ex]
  &\infer[\rname{CP}{Join}]
    {(N_1\ppar N_2),\mu\tr[\phi_1\cup\phi_2](N_1'\ppar N_2'),\sigma'',\mu''\rule[1em]{0em}{0em}}
    {
      N_1,\mu\tr[\phi_1]N_1',\mu'
      &
      N_2,\mu'\tr[\phi_2]N_2',\mu''
    }
\end{eqnarray*}
\caption{Semantics of communications (process level). Side condition $(\dagger)$ in \rname{CP}{Join} is the same as in \rname{C}{Join}.}
\label{fig:proc-compat}
\end{figure}

The remaining rules defining the semantics of CP are standard, and given in Fig.~\ref{fig:cp_semantics}.
Rule $\rname{CP}{Struct}$ uses the structural precongruence relation $\precongr$
(including associativity and commutativity of $\ppar$), defined in the standard way~\cite{CM16b}.

\begin{figure}[t]
\begin{eqnarray*}
&
\infer[\rname{CP}{Cond}]
{
	\proc{\pid p}{\rho}{\cond{e}{B_1}{B_2}}, A
	\toC
	\proc{\pid p}{\rho}{B_i}, A
}
{
	i = 1 \ \text{if } e \eval{\rho}{} \mathsf{true},\quad
	i = 2 \ \text{otherwise}
}
\\[1ex]
&
\infer[\rname{CP}{Ctx}]
{
\proc{\pid p}{v}{\rec{X}{B_2}{B_1}}
\ppar N, A
\toC
\proc{\pid p}{v}{\rec{X}{B_2}{B_1'}}
\ppar N', A'
}{
	\proc{\pid p}{v}{B_1} \ppar N, A
	\toC
	\proc{\pid p}{v}{B'_1} \ppar N', A'
}
\\[1ex]
&\infer[\rname{CP}{Par}]
{
	N \ppar M, A \toC N' \ppar M, A'
}
{
	N,A \ \toC\  N', A'
}
\qquad
\infer[\rname{CP}{Struct}]
{
	N,A \toC N',A'
}
{
	N \precongr M & M,A \toC M',A' & M' \precongr N'
}
\end{eqnarray*}
\caption{Connected Processes, Semantics.}
\label{fig:cp_semantics}
\end{figure}

\subsection{EndPoint Projection (EPP)}
The EPP of a choreography $C$ from CR into CP follows the usual construction, but with an
additional ingredient: we need to add port names associated with communication actions. This is visible in the rules for
projecting the individual behaviour of each process (Fig.~\ref{fig:epp}), notably in the rule for
projecting communications.

\begin{remark}
In Fig.~\ref{fig:epp}, $\port o$ and $\port i$ denote variables that range over concrete ports. Thus, a process $\pid p$ has output ports $\port {o_{\gamma_1}}, \port {o_{\gamma_2}}, \ldots$, and input ports $\port {i_{\gamma_1}}, \port {i_{\gamma_2}}, \ldots$, where $\port o$ and $\port i$ actually stand for $\port {o^{\pid p}}$ (``output port at $\pid p$'') and $\port {i^{\pid p}}$ (``input port at $\pid p$''), while connector $\gamma_1$ knows output ports $\port {o^{\pid p_1}_{\gamma_1}}, \port {o^{\pid p_2}_{\gamma_1}}, \ldots$, and similarly for input ports.
\qed
\end{remark}

\begin{figure}[t]
\small
\begin{eqnarray*}
&
\epp{\tilde\eta\thru\gamma;C}{\pid r} =
	\begin{cases}
		\psend{\port o_\gamma}{e};\epp{C}{\pid r} & \text{if } \pid r = \pid p \text{ and } \gencom \in \tilde\eta \\
		\precv{\port i_\gamma}{x};\epp{C}{\pid r} & \text{if } \pid r = \pid q \text{ and } ( \gencom \in \tilde\eta
\text{ or } \genrecv \in \tilde\eta)\\
		\psel{\port o_\gamma}{\ell};\epp{C}{\pid r} & \text{if } \pid r = \pid p \text{ and } \gensel \in \tilde\eta \\
		\pbranch{\port i_\gamma}{\{ \ell : \epp{C}{\pid r}\}} & \text{if } \pid r = \pid q \text{ and } ( \gensel
\in \tilde\eta
\text{ or } \genchoice \in \tilde\eta)
	\end{cases}
\\[1ex]
&\epp{\gencond}{\pid r} =
	\begin{cases}
		\cond{e}{\epp{C_1}{\pid r}}{\epp{C_2}{\pid r}} & \text{if } \pid r = \pid p \\
		\epp{C_1}{\pid r} \sqcup \epp{C_2}{\pid r} & \pid r\neq\pid p
	\end{cases}
\\[1ex]
&
	\epp{\rec{X}{C_2}{C_1}}{\pid r} = \rec{X}{\epp{C_2}{\pid r}}{\epp{C_1}{\pid r}}
\qquad
\epp{\nil}{\pid r} = \nil
\qquad
\epp{X}{\pid r} = X
\end{eqnarray*}
\caption{Cho-Reo-graphies, Behaviour Projection.}
\label{fig:epp}
\end{figure}

The rule for projecting conditionals uses the standard partial merging operator $\sqcup$, where $B\sqcup B'$
is isomorphic to $B$ and $B'$ up to branching with different labels (see~\cite{CM16b} for details).

We now define the projection of $C$ given a state $\sigma$. As usual, this is the parallel composition of the projections of all processes in $C$.
\[
\epp{C,\sigma}{}  =
\prod_{\pid p \in \pn(C)} \proc{\pid p}{\rho_{\pid p}}{\epp{C}{\pid p}}
\qquad \mbox{ where }\rho_{\pid p}(x)=\sigma(\pid p,x) \mbox{ for each variable $x$ at $\pid p$}
\]
$C$ is \emph{projectable} when $\epp{C,\sigma}{}$ is defined for some $\sigma$. This is equivalent to saying that
$\epp{C,\sigma}{}$ is defined for all $\sigma$.
We illustrate endpoint projection in Ex.~\ref{ex:motivating+impl}
and~\ref{ex:motivating+referee2}.

\begin{example}[Book sale]
\label{ex:motivating+impl}
Continuing with our running example, the choreography presented in
Ex.~\ref{ex:motivating+cr} is projectable, and yields the following network of connected
processes each state $\sigma$.
\setlength{\belowdisplayskip}{0pt}
\begin{align*}
  & \proc{\pid a}{\rho_{\pid a}}{%
    \psend{\port o_{\med{a2c}}}{\mathit{title}};
    \precv{\port i_{\med{c2a}}}{\mathit{price}};
    \cond{(\mathit{happy})}{(%
      \psel{\port o_{\med{a2cbs}}}{\mathit{ok}};
      \psend{\port o_{\med{ac2bs}}}{\mathit{money}};
      \nil
    )\\&\hspace{16.6em}}{(%
      \psel{\port o_{\med{a2cbs}}}{\mathit{ko}};
      \nil
    )}
  } \\
  {}\ppar{}
  & \proc{\pid b}{\rho_{\pid b}}{%
    \pbranch{\port i_{\med{a2cbs}}}{%
      \{
      \mathit{ok}:
      \precv{\port i_{\med{ac2bs}}}{\mathit{money}};
      \nil;\
      \mathit{ko}:
      \nil
      \}
    }
  } \\
  {}\ppar{}
  & \proc{\pid c}{\rho_{\pid c}}{%
    \precv{\port i_{\med{a2c}}}{\mathit{title}};
    \psend{\port o_{\med{c2a}}}{\mathit{price}};
    \pbranch{\port i_{\med{a2cbs}}}{%
      \{
      \mathit{ok}:
      \psend{\port o_{\med{ac2bs}}}{\mathit{book}};
      \nil;\
      \mathit{ko}:
      \nil
      \}
    }
  } \\
  {}\ppar{}
  & \proc{\pid s}{\rho_{\pid s}}{%
    \pbranch{\port i_{\med{a2cbs}}}{%
      \{
      \mathit{ok}:
      \precv{\port i_{\med{ac2bs}}}{\mathit{book}};
      \nil;\
      \mathit{ko}:
      \nil
      \}
    }
  }
\end{align*}
\qed
\end{example}

\begin{example}
\label{ex:motivating+referee2}
It is also worthwhile to note that the following choreographies are
\emph{not} congruent, and that they have \emph{different} EPPs:
\begin{align*}
  C_1 &= \sel{\pid p}{\{\pid q, \pid r\}}{\ell}\thru\gamma;\ \nil
  &
  C_2 &= \sel{\pid p}{\pid q}{\ell}\thru\gamma;\ \sel{\pid p}{\pid r}{\ell}\thru\gamma;\ \nil
\end{align*}
Choreography $C_1$ is syntactic sugar for $\{\sel{\pid p}{\pid
q}{\ell},\sel{\pid p}{\pid r}{\ell}\}\thru\gamma;\nil$. The EPP to $\pid p$,
thus, consists of \emph{one} send; connector $\gamma$ must subsequently ensure
that label $\ell$ is replicated to, and received by, both $\pid q$ and $\pid r$
(i.e., formally, $\gamma$ must be compatible with $C_1$).\footnote{This behaviour
was the motivation for introducing multicasts as abbreviations: the notation in $C_1$ better
conveys how communications really happen; EPP follows this intuition.}

Choreography $C_2$, in contrast, is not congruent to $\{\sel{\pid p}{\pid
q}{\ell},\sel{\pid p}{\pid r}{\ell}\}\thru\gamma;\nil$. Specifically, we cannot
use rule \rname{C}{Eta-Split} to merge the two interactions in $C_2$, because
its disjointness premise does not hold ($\pid p$ occurs in both interactions).
Accordingly, the EPP of choreography $C_2$ on $\pid p$ consists of \emph{two}
sends.
\qed
\end{example}

To state the operational correspondence between a choreography and its projection, we need to map the process names
used as ports in a connector mapping $\conn$ to the actual port names used in networks. We define
$\epp{\conn}{}$ to be the set of all automata in the codomain of $\conn$, where each output port $\pid p$ in automaton
$\conn(\gamma)$ becomes $\pidp po_\gamma$ (and likewise for input ports). We define a similar function for automaton state function
 $\astate$.

\begin{theorem}[Operational Correspondence]
  \label{thm:oc-cc-sp-2}
  Let $C$ be a projectable choreography. Then, for all $\sigma$, $\conn$, and $\astate$:
  \begin{description}
  \item[Completeness:] If $C,\sigma,\astate \redM C',\sigma',\astate'$, then $\epp{C,\sigma}{},\epp{\astate}{} \toCEPP
    \epp{C',\sigma'}{},\epp{\astate'}{}$;
  \item[Soundness:] If $\epp{C,\sigma}{},\epp{\astate}{} \toCEPP N,A'$, then $C,\sigma,\astate \redM
C',\sigma',\astate'$ for some $\sigma'$ and $\astate'$
    with~$\epp{C',\sigma'}{} \prec N$ and $\epp{\astate'}{} = A'$.
  \end{description}
\end{theorem}

\noindent
In the soundness result, the \emph{pruning relation} $\prec$~\cite{CHY12,CM13} states that the
processes in $N'$ may offer more branches than those present in
$\epp{C',\sigma'}{}$, but these are never selected~\cite{CHY12,LGMZ08,DGGLM17}.

In particular, if $\vdash^\conn C$, then $\epp{C,\sigma}{}$ is guaranteed to be deadlock-free when executed with all
automata in $\epp{\conn}{}$ in their initial states.

\begin{example}[Book sale]{}
For any connector mapping $\conn$, the process network in Ex.~\ref{ex:motivating+impl} operates
under $\epp{\conn}{}$ exactly as the choreography in Ex.~\ref{ex:motivating+cr} under $\conn$.
In particular, if $\conn$ respects the original choreography, then this implementation never
deadlocks under $\epp{\conn}{}$.
\qed
\end{example}

\section{Conclusions}
\label{sec:conclusions}

Choreographic approaches to concurrent programming have been heavily
investigated~\cite{Aetal16,Hetal16}, but they
typically adopt some fixed (and restrictive)
communication semantics (like point-to-point synchronous). CR is the first model that modularly integrates choreographies
with what runs ``under the hood'' of communications,
allowing for user-defined communication semantics given as connectors. Thanks to
compatibility (Definition~\ref{defn:compat}), CR inherits the good properties of both Choreographic
Programming and Exogenous Coordination.
Thus, we have significantly extended the applicability of
choreographies: not only can we capture new kinds of behaviours in choreographies (like barriers,
cf.\ Ex.~\ref{ex:motivating+conn}, and alternators, cf.\
Ex.~\ref{ex:motivating+opposite}),
but we can even define systems that integrate different parts with
different communication semantics \emph{and} check whether such integration will lead to deadlocks.
This is essential in many real-world scenarios, where different components with different
communication semantics are usually combined (e.g., some microservices in a distributed system
may asynchronously exchange data to be used later in a synchronous multiparty transition, similarly
to our example).

This work lays the foundations for applying the combined power of choreographies and
connectors to the challenge of concurrent programming, in that
CR contains all the necessary foundations to obtain a concrete implementation.
The results in \S~\ref{sec:processes} specify how to use CR to obtain code in a process model
supported by connectors.
Thus, a natural next step will be to implement CR by combining implementations of
processes generated from choreographies~\cite{chor:website,DGGLM17} with distributed
implementations of Reo
connectors~\cite{jongmans2015partially,proencca2011synchronous,proencca2012dreams}.
The main pieces exist; the main challenge lies in their effective composition, and CR is the
first essential step towards this objective.

CR also provides a very explicit direction for future developments of this new combined research
line: allowing for more kinds of connectors would make the model immediately more
expressive.
By relaxing the requirements we imposed on the automata in CR (see page~\pageref{constraints}), we can introduce non-deterministic
communication semantics to choreographies, to cater to applications that require lossy channels
and safe communication races.
Likewise, a more fine-grained semantics that splits communications into two independent send and
receive actions (similar to~\cite{CM17:ice}) would enrich the class of behaviours
that are captured.

We have followed the traditional approach of viewing choreographies as precise specifications of the intended
interactions. However, it would be reasonable to allow choreographies to underspecify communications, such that the
underlying connectors were allowed to exchange messages also to participants not defined in the choreography. For
example, the semantics for the choreography term $\gencom\thru\gamma$ can allow $\gamma$ to send the message from
$\pid p$ to $\pid q$ via an intermediate process $\pid r$ that may perform additional actions (like logging the
message, or sharing it through another connector).
This generalisation can, in particular, provide a novel way for studying how choreographies can be applied to open-ended systems,
where the processes projected from \emph{multiple} choreographies execute in parallel and share
connectors.

\bibliography{biblio}

\end{document}